\documentclass[twocolumn,pra]{revtex4-2}

\usepackage[dvips]{graphicx} 
\usepackage{amsfonts,amssymb,amscd,amsmath,amsthm}
\usepackage{textcomp}
\usepackage{enumerate}
\usepackage{epsfig}
\usepackage{subfigure}
\usepackage{xcolor}
\usepackage[outdir=./]{epstopdf}
\usepackage[colorlinks = true]{hyperref}
\usepackage{physics}
\usepackage{appendix}
\usepackage{comment, apptools, thmtools, thm-restate}
\usepackage{diagbox}
\usepackage[most]{tcolorbox}

\usepackage{tikz}
\usepackage[]{qcircuit}
\usetikzlibrary{arrows}
\usetikzlibrary{shapes,fadings,snakes}
\usetikzlibrary{decorations.pathmorphing,patterns}
\usetikzlibrary{calc}
\usetikzlibrary{positioning}
\usepackage{etex}
\usepackage{pgfplots}
\usepgfplotslibrary{ternary}
\pgfplotsset{compat=newest}

\newtheorem{theorem}{Theorem}
\newtheorem{lemma}{Lemma}
\newtheorem{corollary}{Corollary}

\newtheorem{proposition}{Proposition}
\newtheorem{definition}{Definition}
\newtheorem{observation}{Observation}

\newtcolorbox[auto counter]{mybox}[2][]{
	enhanced,
	breakable,
	colback=blue!5!white,
	colframe=blue!75!black,
	fonttitle=\bfseries,
	title=Box \thetcbcounter: #2,#1
}

\graphicspath{{./figure/}}

\begin{document}

\title{Intrinsic randomness under general quantum measurements}

\author{Hao Dai}
\affiliation{Center for Quantum Information, Institute for Interdisciplinary Information Sciences, Tsinghua University, Beijing 100084, P.~R.~China}

\author{Boyang Chen}
\affiliation{Center for Quantum Information, Institute for Interdisciplinary Information Sciences, Tsinghua University, Beijing 100084, P.~R.~China}

\author{Xingjian Zhang}
\affiliation{Center for Quantum Information, Institute for Interdisciplinary Information Sciences, Tsinghua University, Beijing 100084, P.~R.~China}

\author{Xiongfeng Ma}
\email{xma@tsinghua.edu.cn}
\affiliation{Center for Quantum Information, Institute for Interdisciplinary Information Sciences, Tsinghua University, Beijing 100084, P.~R.~China}

\begin{abstract}
Quantum measurements can produce randomness arising from the uncertainty principle. When measuring a state with von Neumann measurements, the intrinsic randomness can be quantified by the quantum coherence of the state on the measurement basis. Unlike projection measurements, there are additional and possibly hidden degrees of freedom in apparatus for generic measurements. We propose an adversary scenario for general measurements with arbitrary input states, based on which, we characterize the intrinsic randomness. Interestingly, we discover that under certain measurements, such as the symmetric and information-complete measurement, all states have nonzero randomness, inspiring a new design of source-independent random number generators without state characterization. Furthermore, our results show that intrinsic randomness can quantify coherence under general measurements, which generalizes the result in the standard resource theory of state coherence.
\end{abstract}

\maketitle 


Randomness is an essential resource in cryptography and scientific simulation. Due to its deterministic nature, Newtonian physics fails to provide intrinsically unpredictable randomness. Without such a resource, cryptosystems fail to provide information-theoretic security. Fortunately, the uncertainty principle in quantum physics offers means to generate intrinsic randomness \cite{born1926quantenmechanik}. Thanks to this quantum feature, quantum random number generators (QRNGs) lay down a solid foundation for the security of cryptography systems \cite{ma2016quantum,herrero2017quantum}.

There are various ways to construct a QRNG, typically composed of a source and a detector. The source is characterized by a quantum state \cite{cramer2010efficient,haah2017sample}, while the detector is calibrated by a quantum measurement \cite{d2004quantum,lundeen2009tomography}. After obtaining outcomes from the quantum measurement, the legitimate user, Alice, needs to analyse the amount of randomness from the raw data. This analysis can be put in an adversary scenario, as shown in Figure~\ref{fig:povm}. The adversary, Eve, might have a certain correlation with the system. Such correlation could leak the information of measurement outcomes to Eve. To remove the information leakage, we can divide the entropy of outcomes into two parts: intrinsic randomness, about which Eve has no information, and extrinsic randomness, which Eve might know. The essential task in randomness analysis is to quantify the intrinsic randomness given an input state and a measurement. Note that in the (semi-)device-independent scenarios, Alice might skip some of these steps. For example, in source-independent schemes \cite{Cao2016Source,Marangon2017Source}, the source is assumed to be uncalibrated or even untrusted.

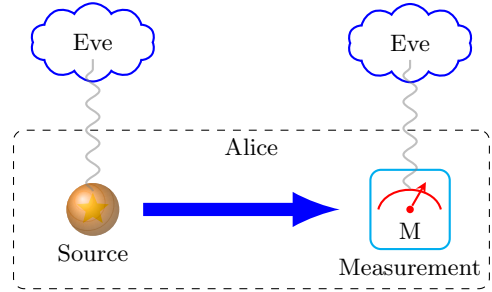
\begin{figure}[tbhp!]
	\begin{tikzpicture}[
		scale=1,
		]
		\node[star, star points=5, star point ratio=2.25, inner sep=2pt, fill=yellow] (S) at (0,0) {};
		\node[shade,shading=ball, circle, ball color=orange, minimum size=20pt, opacity=0.3,label=below:Source] at (S) {};
		\node[circle,fill,red,inner sep=1pt,label=below:M] (M) at (120pt,0pt) {};
		\node[draw,rectangle,rounded corners,minimum size=30pt,cyan,thick,label=below:Measurement] at (M) {};
		\draw (M) [xshift=-12pt,red,thick] arc (170:10:12pt and 9pt);
		\draw[-stealth,yshift=0pt,red,thick]  (M) --++ (60:12pt);
		\draw[-latex,shorten <=15pt,shorten >=25pt, blue, line width=5pt] (S.east) -- (M.west);
		\node[draw,rectangle,rounded corners,minimum height=60pt, minimum width=180pt, black,dashed, label={[anchor=north]north:Alice}] at (60pt,0) {};
		\node[above=50pt of S] (Es) {Eve};
		\node[cloud, draw,cloud puffs=10, cloud puff arc=120, aspect=2, inner ysep=8pt,blue,thick] at (Es) {};
		\draw[-,decorate,decoration=snake,thick,gray,opacity=0.5] (S.north) -- (Es.south);
		\node[above=50pt of M,yshift=5pt] (Em) {Eve};
		\node[cloud, draw,cloud puffs=10, cloud puff arc=120, aspect=2, inner ysep=8pt,blue,thick] at (Em) {};
		\draw[-,decorate,decoration=snake,thick,gray,opacity=0.5] (M.north)+(up:6pt) -- (Em.south);
	\end{tikzpicture}
\caption{Illustration of a typical QRNG. The source sends quantum signals in the state of $\rho$ to the measurement device, which outputs a sequence of random numbers. Eve could have a certain correlation with the devices, where she could possess the purification of $\rho$ on the source side and know the construction of the detection on the measurement side. Eve might even have entanglement with the internal apparatus.} \label{fig:povm}
\end{figure}


As for the intrinsic randomness quantification, let us start with a well-studied special case with the detection calibrated as a von Neumann measurement, $\{\ketbra{i}\}$. If the input is in a superposition state, $|\psi\rangle=\sum_{i}a_i |i\rangle$ with normalized complex coefficients $a_i\in\mathbb{C}$ and $\sum_i\abs{a_i}^2=1$, the measurement outcome is intrinsically random and the probability of obtaining outcome $i$ is $|a_i|^2$ according to Born's rule \cite{born1926quantenmechanik}. In this case, intrinsic randomness of the outcomes arises from breaking superposition \cite{zurek2009quantum} and is given by the Shannon entropy of the probability distribution, $\{|a_i|^2\}$. In resource theory, superposition is quantified by quantum coherence with respect to the measurement basis, $\{\ket{i}\}$ \cite{Aberg2006,baumgratz2014quantifying}. In fact, for a generic input state described by a density matrix, the link between output intrinsic randomness and state coherence has been established \cite{yuan2015intrinsic,Hayashi2018secure,yuan2019quantum}.

A projection measurement is an idealized model for detection devices. In reality, noise is inevitable, or equivalently, part of instrument information is missing from the user's point of view. Then, the detection is generally characterized by a positive-operator-valued measure (POVM). How to quantify intrinsic randomness of the outcomes from a generic measurement is an important yet unsettled problem. Given a set of POVM elements, as many degrees of freedom in measurement instruments are hidden from Alice, there is an infinite number of ways to construct the detection instrument \cite{duvsek2002quantum,peres1990neumark,biggerstaff2009cluster}. This hidden information makes it very challenging to characterize the amount of information leaked to Eve. In the literature, there have been some attempts on this topic \cite{law2014quantum,cao2015loss,bischof2017measurement}. For instance, one can express a generic measurement as a mixture of projection measurements followed by classical postprocessing \cite{cao2015loss}. Unfortunately, this mixing technique is only feasible for measurements with two outcomes in a qubit system while failing in general cases \cite{Oszmaniec2017}.

Following the spirit of studying randomness for projection measurements from the perspective of coherence, we find that the existing coherence measures under POVMs cannot properly quantify intrinsic randomness \cite{Luo2017quantum,Bischof2019,xu2020coherence,bischof2021quantifying}. Let us illustrate this with an example. Consider the two-outcome POVM $\mathbf{M}=\{\mathbf{1}/2,\mathbf{1}/2\}$, which is free in the resource theory of measurement informativeness \cite{skrzypczyk2019robustness}. The measurement outcome is independent of input states and can be seen as a classical random variable taking values $0$ and $1$ with an equal probability, which we expect to be of a classical nature. However, all states have nonzero coherence in the definition given by the direct application of conventional Naimark extension \cite{Bischof2019}.

In this work, we provide a generic adversary scenario where the detection is correlated with Eve as shown in Figure~\ref{fig:povm}. We take all the hidden variables or missing information as an ancillary system into consideration, where the POVM can be viewed as a part of projection on a larger system. This is a generalized version of Naimark extension, with the difference that the ancillary state is not necessarily pure as in the conventional one \cite{neumark1943spectral,peres1990neumark}. Then, we can apply the results of intrinsic randomness quantification for projection measurements. As the Naimark extension is not unique, we need to minimize over all possible extensions in the randomness analysis.


For a special type of measurements, extremal POVM \cite{d2005classical}, we show that for all Naimark extensions with possibly mixed ancillary states, the randomness function under an extremal POVM is equal and thus gives the intrinsic randomness. Surprisingly, for some extremal measurements, such as the symmetric and information-complete (SIC) measurement \cite{renes2004symmetric,fuchs2017sic}, their outcomes have non-zero randomness for any input states. Then, we can design a new source-independent QRNG using these measurements.
Moving on to a general POVM, the mixed ancillary state in the Naimark extension is formidable. We make an additional assumption that Eve performs a measurement on her local system and the ancillary state becomes a particular mixture of pure states, from which a convex-roof construction of the intrinsic randomness is obtained.


Moreover, we regard the randomness quantification as a state coherence measure under POVMs. Following a standard resource-theoretic approach, we define a set of incoherent states and incoherent operations for a general measurement. In our coherence measure, for POVM $\{\mathbf{1}/2,\mathbf{1}/2\}$, all states are incoherent.

\textit{Randomness characterization for general POVMs.}— 
For a $d$-dimensional Hilbert space $ \mathcal{H} =\mathbb{C}^d,$ a POVM on  $ \mathcal{H} $ is a set of positive semidefinite Hermitian operators $\mathbf{M}=\{M_{1},\cdots,M_{m}\} $, where $\sum_{i=1}^{m} M_{i}=\mathbf{1}$.
When two POVMs are the same, $\forall i$, $M_i=N_i$, we denote by $\mathbf{M}=\mathbf{N}$. Each element can be expressed as $  M_{i}= A_{i} A_{i}^{\dagger} $, where $A_{i}$ is called a POVM operator and generally not a square matrix. When measuring a state $\rho$, the probability of obtaining the outcome $i$ is given by $\tr(M_{i}\rho)$ and the corresponding post-measurement state is $ A_{i}\rho A_{i}^{\dagger} / \tr(M_{i}\rho)$. The set of operators $\{A_{i}\}$ uniquely determines the implementation of the measurement --- instrument. On the other hand, a POVM generally corresponds to many possible implementations or different sets of operators, $\{A_{i}\}$.

The projection measurement, also called projection-valued measure (PVM), is a special case of a POVM when $M_i$ are projection operators, $M_i^2=M_i$, and $M_i=A_i$. The post-measurement states of PVMs are unique, thus can be regarded as basic building blocks in the implementation of POVMs. As a special case, when every PVM element is rank-1, we call it von Neumann measurement.

As the first step of randomness evaluation, Alice characterizes the source in a QRNG to be state $\rho^A$ and calibrates the detection device to be measurement $\mathbf{M}$. By introducing an ancillary system $Q$, we can extend $\mathbf{M}$ to a PVM, $\mathbf{P}$. In the conventional Naimark extension, the ancillary state of $Q$ is assumed to be pure. Here, we generalize it for a mixed state $\sigma^Q$. In the adversary picture, both the source and the measurement could be correlated with Eve. In the worst-case scenario, Eve holds the purification of the source state $\rho^A$ and the ancillary state $\sigma^Q$, as shown in Figure~\ref{fig:Extension}. The measurement extension requires $\mathbf{M},\mathbf{P}$ and $\sigma^Q$ to satisfy the consistency condition \cite{busch1996quantum}, where $\forall i$,
\begin{equation}\label{eq:consistency}
M_i= \tr_{Q}[P_i(\mathbf{1}^{A}\otimes\sigma^{Q})].
\end{equation}
That is, $\forall \rho^A, i$, $\tr(M_i\rho^A)=\tr[P_i(\rho^A\otimes \sigma^Q)]$. For simplicity, we shall omit the superscript $A$ and $Q$ when there is no confusion in the following discussions and denote the states $\rho$ and $\sigma$ as Alice's state and the ancillary state, respectively. Once Eve gives the measurement device to the user, she cannot access the apparatus anymore. Hence, there is a no-signalling relation between the input state and Eve's system.

\begin{figure}[hbtp!]
	\begin{tikzpicture}[
		scale=1,
		]
		\node[circle,fill,gray,inner sep=1pt,label=left:$A$] (A) at (0,0) {};
		\node[circle,fill,gray,inner sep=1pt,label=left:$Q$] (Q) at ($ (A) + (1,-1) $) {};
		\node[circle,fill,gray,inner sep=1pt,label=left:$E$] (E) at ($ (A) + (0,1) $) {};
		\node[circle,fill,gray,inner sep=1pt,label=left:$F$] (F) at ($ (A) + (1,-2) $) {};
		\draw[snake=snake,gray] (A) -- (E) node[midway,left] {$\ket{\Psi}^{EA}$};
		\draw[snake=snake,gray] (Q) -- (F) node[midway,left] {$\ket{\Phi}^{QF}$};
		\node[draw,black,fill=blue!10!white,minimum height=1.8cm,rounded corners] (M) at ($ (A) + (2.2,-.5) $) {PVM};
		\draw (A) -- (A-|M.west) node[near start,above] {$\rho^A$} (Q) -- (Q-|M.west) node[midway,above] {$\sigma^Q$};
		\draw[-stealth,double] (M.east) -- ++ (1cm,0) node[near start,above] {};
		\draw[-stealth] (E) -- ($(E-|M.east) + (1cm, 0) $);
		\draw[-stealth] (F) -- ($(F-|M.east) + (1cm, 0) $);
		\draw[orange,dashed,rounded corners] ($ (A-|Q)!0.5!(E) $) rectangle ($(Q-|M.east)!0.5!(F-|M.east)+(.5cm,0)$) node{$\mathbf{M}$};
	\end{tikzpicture}
\caption{The adversary scenario for a generalized Naimark extension: a PVM on the joint system $AQ$. From Alice's perspective, the measuring process is described by POVM $\mathbf{M}$, depicted as the dashed box. Alice inputs state $\rho^A$ and obtains classical outputs from the box. The ancillary system is generally in a mixed state $\sigma^Q$. If $\sigma^Q$ is pure, it becomes the conventional Naimark extension. Both the source and the ancillary system could be entangled with Eve. There is no-signalling  from input $A$ to Eve's purification $F$.} \label{fig:Extension}
\end{figure}
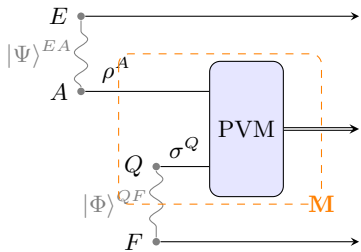

In general, any implementation of a measurement can be treated as a part of a PVM on a larger system. Then, we can quantify intrinsic randomness of POVM outcomes via its extension,
\begin{equation}\label{eq:minR}
\begin{split}
R(\rho,\mathbf{M}) &= \min_{\mathbf{P},\sigma} R(\rho\otimes \sigma,\mathbf{P}), \\
\textrm{s.t.} \quad  \forall i, \;  &M_i={\rm tr}_{Q}[P_i(\mathbf{1}^{A}\otimes\sigma^{Q})], \\
\end{split}
\end{equation}
where the constraint is given by the consistency condition in Eq.~\eqref{eq:consistency}. For the special case of von Neumann measurements, the randomness quantification $R(\varrho,\mathbf{P})$ is well studied in the literature under different adversary scenarios \cite{Devetak2005distillation,yuan2019quantum}. As the starting point of our randomness quantification of POVMs, we generalize the results to the general PVM case and give two widely used measures,
\begin{equation}\label{eq:tworanfun}
\begin{split}
R_c(\varrho,\mathbf{P})&=\underset{\{q_j,\ket{\psi_j}\}}{\min}\sum_j q_j S\left(\Delta_{\mathbf{P}}(\ketbra{\psi_j})\right),\\
R_q(\varrho,\mathbf{P})&=S\left(\varrho\|\Delta_{\mathbf{P}}(\varrho)\right),
\end{split}
\end{equation}
where $S(\varrho)=-\tr(\varrho\log\varrho)$ represents the von Neumann entropy function, $\varrho=\sum_j q_j\ketbra{\psi_j}$ is a state decomposition, and $\Delta_{\mathbf{P}}(\varrho)=\sum_{i=1}^{m}P_{i}\varrho P_{i}$ is the block-dephasing operation. The difference between $R_c$ and $R_q$ lies in whether Eve performs measurements on each copy of system $E$ in Figure \ref{fig:Extension} and the two functions coincide when $\varrho$ is pure. Technical details here and below are presented in Supplemental Material \cite{suppmaterial}.

Since a classical mixture of quantum states should not increase the output intrinsic randomness on average, the randomness function $R(\varrho, \mathbf{P})$ satisfies the convexity condition,
\begin{equation}\label{eq:RPconvex}
	\begin{split}
R\left(\sum_j r_j \varrho_j, \mathbf{P}\right) &\le \sum_j r_j R(\varrho_j, \mathbf{P}),
	\end{split}
\end{equation}
for arbitrary coefficients $\sum_j r_j=1$ and $r_j\geq0$. A state, $\varrho=\sum_j r_j \varrho_j \equiv\oplus_jr_j \varrho_j$, is \textit{block-diagonal} with respect to $\mathbf{P}$, when $\forall j\neq j'$ and $\forall P_{i}\in \mathbf{P}$, there is $\tr(P_{i}\varrho_jP_{i}\varrho_{j'})=0 $. Intuitively, different block subspaces should have no interference with each other when measuring. This indicates that the randomness function should also satisfy the additivity condition for the block-diagonal states,
\begin{equation}\label{eq:additivity}
	\begin{split}
		R(\oplus_j r_j \varrho_j, \mathbf{P}) &= \sum_j r_j R(\varrho_j, \mathbf{P}).
	\end{split}
\end{equation}

Note that, the two aforementioned randomness functions for PVMs meet these criteria. With the convexity condition on the randomness function, we can see that the randomness defined in Eq.~\eqref{eq:minR} satisfies the convexity condition for both $\rho$ and $\mathbf{M}$. From the resource theory point of view, these two conditions stem from the convexity \cite{baumgratz2014quantifying} and the additivity on block-diagonal states \cite{Yu2016Alternative} of coherence measures.



Let us check out a special case where the POVM is extremal, which cannot be decomposed into a linear mixture of other POVMs \cite{d2005classical}. This is an analog to a pure state, which is often considered to be decoupled from the environment. An extremal POVM can also be treated as a measurement decoupled from the environment. In the adversary scenario, there is no hidden variable for a pure state or an extremal POVM. That is, in Figure \ref{fig:Extension}, system $F$ is trivial. To put this intuition in a rigorous manner, we show that for an extremal POVM, the intrinsic randomness is independent of the extension $\{\mathbf{P},\sigma\}$.


\begin{theorem}\label{thm:ExtRCan}
For an extremal POVM $\mathbf{M}$ and a fixed input state $\rho$, all the generalized Naimark extensions give the same amount of randomness.
\end{theorem}

Then, we can skip the minimization problem in Eq.~\eqref{eq:minR} and employ any extension for the randomness function. In practice, we can take a canonical extension of $\mathbf{M}$ \cite{peres1990neumark}, denoted by $\mathbf{P}_c$,
\begin{equation}\label{eq:RexCanon}
	\begin{split}
		R(\rho,\mathbf{M}) &= R(\rho\otimes \ketbra{0},\mathbf{P}_c). \\
	\end{split}
\end{equation}

A general POVM can be decomposed to extremal ones, just like that a mixed state can be decomposed to pure states. The decomposition of a POVM is generally not unique, which is controlled by a hidden variable from Alice's point of view. In the generalized Naimark extension as shown in Figure \ref{fig:Extension}, the following Proposition connects the decomposition of the POVM with that of the ancillary state.



\begin{proposition}[Correspondence between ancillary state and measurement decomposition]\label{Pro:CorrDecomp}
In a generalized Naimark extension of a POVM, $\mathbf{M}$, if the ancillary state has a pure state decomposition, $\sigma=\sum_j r_j \ketbra{\varphi_j}$ with $\sum r_j=1$ and $r_j>0$, then there exists a measurement decomposition  $\mathbf{M}=\sum_{j} r_j \mathbf{N}^j$, and vice versa.
\end{proposition}

If Eve performs a local measurement on her system $F$, without loss of generality, the measurement can be restricted to be rank-1. Otherwise, the measurement can be viewed as a rank-1 measurement followed by coarse graining. Then, the ancillary state is chosen from a pure state ensemble and the POVM degenerates into a mixture of corresponding POVMs according to Proposition \ref{Pro:CorrDecomp}. The intrinsic randomness of Alice's outcomes is a weighted average of the randomness for each pure input ancillary state. So, the minimization problem of Eq.~\eqref{eq:minR} becomes minimizing the value $\sum_j r_j R(\rho\otimes\ketbra{\varphi_j}, \mathbf{P})$ over all possible Naimark extensions and pure state decompositions of the ancillary state.

Denote the solution to the minimization problem after Eve's measurement to be $\mathbf{P}^*$ and $\sigma^*=\sum_j r_j^* \ketbra{\varphi_j^*}$. We show that the corresponding measurement decomposition, $\mathbf{M}=\sum_{j} r_j^* \mathbf{N}^{*j}$, is extremal --- indicating that $\{\mathbf{N}^{*j}\}$ are all extremal. The intrinsic randomness for the POVM outcomes is given by $\sum_j r_j^* R(\rho, \mathbf{N}^{*j})$. As a result, we can minimize over all possible extremal decompositions for the POVM to evaluate Eq.~\eqref{eq:minR} and give a convex-roof construction of intrinsic randomness, as presented in the following theorem.

\begin{theorem}\label{thm:POVMrandom}
When Eve performs a measurement on her system $F$, the intrinsic randomness of POVM outcomes is given by,
\begin{equation}\label{eq:Rexdecomp}
	\begin{split}
R^{cf}(\rho,\mathbf{M}) &= \min_{\{\mathbf{N}^j,r_j\}} \sum_j r_j R(\rho, \mathbf{N}^{j}), \\
\textrm{s.t.} \quad  & \mathbf{M}=\sum_{j} r_j \mathbf{N}^j, \\
	\end{split}
\end{equation}
where the decomposed POVMs $\{\mathbf{N}^{j}\}$ are all extremal and the randomness function $R(\rho, \mathbf{N}^{j})$ is given by Eq.~\eqref{eq:RexCanon}.
\end{theorem}

Similar to the case of pure states, when a POVM $\mathbf{M}$ is extremal, there is $R(\rho,\mathbf{M})=R^{cf}(\rho,\mathbf{M})$.

After quantifying randomness for the measurement outcomes with respect to a given POVM, it is interesting to consider the set of states that have no randomness, called \textit{non-random state}. For the special case of a von Neumann measurement, a non-random state is diagonal in the measurement basis \cite{baumgratz2014quantifying,yuan2015intrinsic}. For a general PVM, a non-random state has a pure-state decomposition such that each decomposed state is a $+1$ eigenstate of a measurement projector. Here, we give necessary and sufficient conditions for the non-random states under a generic measurement in the following two corollaries.


\begin{corollary}[Necessary and sufficient condition for non-random states]\label{co:iffnonrandom}
Given a POVM $\mathbf{M}$, a state $\rho$ is non-random, $R^{cf}(\rho,\mathbf{M})=0$, iff the measurement has an extremal decomposition, $\mathbf{M}=\sum_{j}r_j \mathbf{N}^{j}$, satisfying one of the following two equivalent conditions,
\begin{enumerate}
\item
$\forall j, i\neq i'$, $N_{i}^{j}\rho N_{i'}^{j}=0$;

\item
for each $\mathbf{N}^j$, the state has a corresponding spectral decomposition, $\rho=\sum_{k} q_k^j \ketbra{\psi_k^j}$, such that $\forall k$, $N_{i}^{j}\ket{\psi_k^j}=\ket{\psi_k^j}$ for some element $N_{i}^{j}$.
\end{enumerate}
\end{corollary}

For the special case of pure state $\ket{\psi}$, we can derive the necessary and sufficient condition for the general randomness quantification given in Eq.~\eqref{eq:minR}.
\begin{corollary}\label{co:commoneigen}
Given a POVM $\mathbf{M}$, a pure state $\ket{\psi}$ is non-random, $R(\ketbra{\psi},\mathbf{M})=0$, iff $\ket{\psi}$ is a common eigenstate of all measurement elements.
\end{corollary}

For extremal POVMs, according to Corollary \ref{co:iffnonrandom}, the necessary and sufficient condition for zero randomness is that $ \rho $ has a corresponding spectral decomposition, $\sum_{k} q_k \ketbra{\psi_k}$, such that each term $\ket{\psi_k}$ is a $+1$ eigenstate of some element. Intriguingly, there exist particular extremal POVMs, all measurement elements do not have +1 eigenvalue. For example, the SIC measurement is extremal and composed of $ d^{2} $ rank-1 operators, $\{\ketbra{\phi_{i}}/d\}$, with normalized vectors $\ket{\phi_{i}}$ satisfying, $\forall i\neq j$, $\abs{\braket{\phi_{j}}{\phi_{i}}}^{2}=\frac{1}{d+1}$. Each POVM element only has 0 or $1/d$ as eigenvalues. Hence, there is no non-random states for the SIC measurement.

\begin{observation}\label{obs:nononrandom}
For some POVMs, such as SIC measurements, there does not exist a non-random state.
\end{observation}
This observation can help us design source-independent QRNGs. Given a calibrated measurement, if Alice figures that non-random states do not exist, she can be sure that there is positive amount of randomness in the outcomes even without any source characterization. In this case, the lower bound of outcome randomness is given by
\begin{equation}\label{eq:SIRlower}
	\begin{split}
		 R(\mathbf{M}) &= \min_{\rho} R(\rho, \mathbf{M}), \\
	\end{split}
\end{equation}
where the randomness function $R(\rho, \mathbf{M})$ is given in Eq.~\eqref{eq:minR}. This kind of source-independent QRNG designs is stronger than the existing ones \cite{Cao2016Source,Marangon2017Source}, where at least partial source tomography is required. This observation can also help us design other device-independent QRNGs \cite{tavakoli2021mutually}.

\begin{theorem}\label{th:lbousic}
For a SIC measurement $\mathbf{M}$, a lower bound of intrinsic randomness is given by,
\begin{equation}
 R(\mathbf{M})>\log(\frac{d+1}{2}),
\end{equation}
where $d$ is the dimension of the corresponding space.
\end{theorem}

It is worth mentioning that for a SIC measurement with an arbitrary input state, the min-entropy given by the guessing probability is lower bounded by $2\log d-1$ \cite{avesani2020unbounded}. This min-entropy based work only considers the case where Eve does not know the ancillary state of detection devices.

\textit{Intrinsic randomness as a POVM coherence measure.}—
There is a close relation between randomness and coherence for PVMs. In Supplemental Material, we show that for a general PVM, the intrinsic randomness of the outcomes can identify the quantum coherence of the input state with respect to the PVM. Following this spirit, it is natural to regard the intrinsic randomness $R$ in Eq.~\eqref{eq:minR} as a coherence measure for a general POVM. Define the set of non-random states to be the set of POVM-incoherent states,
\begin{equation}
\mathcal{I}_{\mathbf{M}}=\{\rho \mid R(\rho,\mathbf{M})=0\}.
\end{equation}
Corollary \ref{co:commoneigen} gives the characterization of pure POVM-incoherent states. If we take the convex-roof construction in Theorem \ref{thm:POVMrandom}, $R^{cf}(\rho,\mathbf{M})$, Corollary \ref{co:iffnonrandom} would give a full description of this POVM-incoherent state set. The set of POVM-incoherent states, $\mathcal{I}_{\mathbf{M}}$, is convex and can be empty for some special POVMs. Unlike the set of POVM-incoherent states which can be empty, incoherent operations always exist for any POVM. For example, the identity map is a trivial POVM-incoherent operation. Here, we give a definition of the POVM-incoherent operations.

\begin{definition}\label{def:incooper}
For POVM $\mathbf{M}$, operation $\Lambda$ is called incoherent, if for any generalized Naimark extension $\{\mathbf{P},\sigma\}$ on a larger space $\mathcal{H}'$, $\Lambda$ has a corresponding extended operation $\Lambda'$ on $\mathcal{H}'$ that satisfies the following two conditions,
\begin{itemize}
\item
$\Lambda'(\rho\otimes\sigma)=\Lambda(\rho)\otimes\sigma$;

\item
$K_i' \mathcal{I}_{\mathbf{P}} K_i'^{\dagger} \subseteq \mathcal{I}_{\mathbf{P}}$ where $\mathcal{I}_{\mathbf{P}}$ is the set of incoherent states of $\mathbf{P}$, and $K_i'$'s are the Kraus operators of $\Lambda'$.
\end{itemize}
\end{definition}


Under this definition, the coherence measures $C$ defined by $R$ and $R^{cf}$ both satisfy the following essential criteria.
(i) Nonnegativity: $C(\rho,\mathbf{M})\ge 0$, and $C(\delta,\mathbf{M})=0$ iff $\delta\in \mathcal{I}_{\mathbf{M}}$;
(ii) Monotonicity: for any POVM-incoherent operation $\Lambda$, $ C(\Lambda(\rho),\mathbf{M})\le C(\rho,\mathbf{M})$;
(iii) Strong monotonicity: for any POVM-incoherent operation $\Lambda$ with Kraus operators $K_i$, $\sum_i p_i C(\rho_i,\mathbf{M})\le C(\rho,\mathbf{M})$ with $p_i=\tr(K_i\rho K_i^{\dagger})$ and $\rho_i=K_i\rho K_i^{\dagger}/p_i$;
(iv) Convexity: $C(\sum_{j}\lambda_{j}\rho_{j},\mathbf{M})\leq \sum_{j}\lambda_{j} C(\rho_{j},\mathbf{M})$ with $\lambda_{j}>0$ and $\sum_j \lambda_{j}=1$.
Besides, in the special case of PVMs, the definitions for coherence measures, incoherent states, and incoherent operations are identical with their corresponding part in the block-coherence theory.
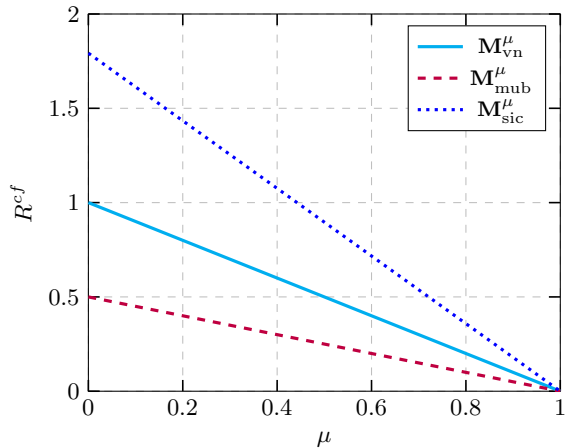
\begin{figure}
\begin{tikzpicture}
\pgfplotsset{every axis/.append style={solid},
        every tick/.append style={semithick,color=black},
    }
    \begin{axis}[smooth,
    scale only axis = true, width = 0.35\textwidth, height = 0.28\textwidth,
    ymin =0, ymax =2, xmin =0,  xmax=1,
    ylabel = {$R^{cf}$},
    xlabel={$\mu$},
    grid style ={dashed},
    grid = both,
    legend style={font=\footnotesize,
    legend pos =north east},
    ]
\addplot [cyan,very thick,solid, domain=0:1] {(1-x)*1};
\addlegendentry{$\mathbf{M}^{\mu}_{\mathrm{vn}}$};

\addplot [purple,very thick,dashed, domain=0:1] {(1-x)*0.5};
\addlegendentry{$\mathbf{M}^{\mu}_{\mathrm{mub}}$};

\addplot [blue,very thick,dotted, domain=0:1] {(1-x)*0.5*log2(3)+1-x};
\addlegendentry{$\mathbf{M}^{\mu}_{\mathrm{sic}}$};
\end{axis}
\end{tikzpicture}
\caption{Comparison among coherence measures of the qubit state $\ket{0}$ with respect to three specific POVMs.} \label{fig:coherence}
\end{figure}

Note that the coherence measures $R^{cf}$ derived from two randomness functions given in Eq.~\eqref{eq:tworanfun}, $R^{cf}_{c}$ and $R^{cf}_{q}$, can both be efficiently evaluated, see, Supplemental Material \cite{suppmaterial}. For pure input states, $R^{cf}_{c}=R^{cf}_q$ and hence can be briefly written as $R^{cf}$. As an example, we calculate the coherence measure $R^{cf}$ of the qubit state $\ket{0}$ with respect to three POVMs, $\mathbf{M}^{\mu}_{\mathrm{vn}}$, $\mathbf{M}^{\mu}_{\mathrm{mub}}$, and $\mathbf{M}_{\mathrm{sic}}^{\mu}$. When $\mu=0$, these three POVMs are
\begin{equation}
\begin{split}
\mathbf{M}_{\mathrm{vn}}&=\{\ketbra{+},\ketbra{-}\},\\
\mathbf{M}_{\mathrm{mub}}&=\left\{\frac12\ketbra{0},\frac12\ketbra{1},\frac12\ketbra{+},\frac12\ketbra{-}\right\},\\
\mathbf{M}_{\mathrm{sic}}&=\left\{\frac12\ketbra{0},\frac12\ketbra{\phi_0},\frac12\ketbra{\phi_1},\frac12\ketbra{\phi_2}\right\},
\end{split}
\end{equation}
where $\ket{\pm}=(\ket{0}\pm\ket{1})/\sqrt{2}$ and $\ket{\phi_k}=\sqrt{1/3}\ket{0}+{\rm e}^{2k\pi i/3}\sqrt{2/3}\ket{1}$, $k=0,1,2$.
When $0<\mu\le 1$, $\mathbf{M}^{\mu}_{\mathrm{vn}}=(1-\mu)\mathbf{M}_{\mathrm{vn}}+\mu\mathbf{I}_2$ with $\mathbf{I}_2=\{a\mathbf{1},(1-a)\mathbf{1}\},a\ge0$. Then, from the definition, we can see that $R^{cf}(\ket{0},\mathbf{M}^{\mu}_{\mathrm{vn}})=(1-\mu) R^{cf}(\ket{0},\mathbf{M}_{\mathrm{vn}})$. There are similar definitions and results for $\mathbf{M}^{\mu}_{\mathrm{mub}}$ and $\mathbf{M}^{\mu}_{\mathrm{sic}}$. The state coherence under these three POVMs are compared in Figure~\ref{fig:coherence}. We can see $R^{cf}(\ket{0},\mathbf{M}^{\mu}_{\mathrm{sic}})$ maintains larger than $1$ even with strong deterministic noise, while $R^{cf}(\rho,\mathbf{M}^{\mu}_{\mathrm{mub}})\le 1$ and $ R^{cf}(\rho,\mathbf{M}^{\mu}_{\mathrm{vn}})\le 1$ hold for arbitrary states. This demonstrates the advantage of the SIC measurement in design of QRNG.



\textit{Discussion.}—
In this work, we characterize intrinsic randomness for general states under general measurements. We conjecture that the solution to the minimization problem in Eq.~\eqref{eq:minR} is a fixed Naimark extension in a finite-dimensional space, independent of the state $\rho$.


Our results also shed light on the information-theoretic analysis of quantum measurement processes \cite{buscemi2008global,luo2010information,buscemi2014noise}. In quantum mechanics, we know that any measurement that extracts information out of a state would inevitably introduce disturbance. The global information balance relation claims that there is a gap between the introduced disturbance and the information gain, as information leakage \cite{buscemi2008global}. The relation between intrinsic randomness and the information gain is left for future research.

\acknowledgments
We acknowledge Junjie Chen, Zhenhuan~Liu, Xiao Yuan, and Pei Zeng for the insightful discussions. This work is supported by the National Natural Science Foundation of China Grants No.~11875173 and No.~12174216, the National Key Research and Development Program of China Grants No.~2019QY0702 and No.~2017YFA0304004.

%

\newpage
\begin{appendix}
\onecolumngrid
\setcounter{table}{0}   
\setcounter{figure}{0}
\setcounter{definition}{0}
\setcounter{theorem}{0}
\setcounter{proposition}{0}
\setcounter{lemma}{0}
\setcounter{corollary}{0}
\begin{center}
\textbf{Supplemental Material: Intrinsic randomness under general quantum measurements}
\end{center}

In this supplemental material, we complement the main text with technical details. In Section \ref{App:BlockCoherence}, we derive the intrinsic randomness for general PVMs and its some useful properties. In Section \ref{App:Generalizedextension}, we present the generalized Naimark extension from the perspective of preprocessing and prove Proposition \ref{appPro:CorrDecomp}. In Section \ref{App:ExtremalPOVM}, we review the definition and a property of extremal POVM and prove Theorem \ref{appthm:ExtRCan}. In Section \ref{App:Intrinsicrandomness}, we prove Theorem \ref{appthm:POVMrandom}, Corollary \ref{appco:iffnonrandom}, Corollary \ref{appco:commoneigen}, and Theorem \ref{appth:lbousic}, and present a numerical approach for the intrinsic randomness function.


\begin{table}[h!]
  \begin{center}
    \caption{NOTATION}
    \begin{tabular}{l|l} 
      \textbf{notation} & \textbf{description} \\
      \hline
      $\mathcal{H} $& Hilbert space  \\
      $\mathbf{M}$ &  POVM \\
      $M_{i} $& element of POVM $\mathbf{M}$   \\
      $ \mathbf{P}$  &  PVM \\
      $\Delta_{\mathbf{P}}$ & block-dephasing operation defined by $\mathbf{P}$ \\
      $\mathcal{I}_{\mathbf{P}}$ & the set of block-incoherent states \\
      $H  $ & Shannon entropy\\
      $S  $ & von Neumann entropy\\
      $R  $ & randomness function\\
      $[m]$ & the set $\{1,\cdots,m\}$\\
      $ \mathcal{P}(m) $&the set of POVMs with at most $ m $ outcomes\\
      $ \mathcal{P}$&the set of POVMs with discrete outcomes\\
      $\{\mathbf{P},\sigma\}$ & generalized Naimark extension\\
      $\log$ & logarithm based 2 \\
      $\rho$ & state to be measured and held by Alice $\rho^A$ \\
    \end{tabular}
  \end{center}
\end{table}

\clearpage

\section{Block randomness and block coherence}\label{App:BlockCoherence}
Here, we calculate the intrinsic randomness under projection-valued measures (PVMs), prove several useful properties and link randomness with block coherence. Before discussing the general PVMs, we briefly review randomness and coherence for the special case of rank-1 PVMs --- von Neumann measurement $\{\ketbra{i}\}$, where the intrinsic randomness via measuring a state $\rho$ is well-studied in the literature. For a pure state $\ket{\psi}$, the randomness of its measurement outcomes is directly given by Born's rule, $S(\Delta_{\{\ketbra{i}\}}(\ketbra{\psi}))$, where $S(\rho)=-\tr(\rho\log\rho)$ is the von Neumann entropy function.
To extend the result to the general cases of mixed states, we consider an adversary scenario, where we assume an adversary, Eve, holds the purification system, $\ket{\Psi}^{AE}$. For simplicity,  we consider the case where Alice inputs the same state independently for many rounds of measurement, which corresponds to the independent and identically distributed (i.i.d.) case in the asymptotic limit. In quantum information theory, there are two ways to quantify the amount of randomness by measuring a mixed state.
In the first case, one can take the convex-roof construction of the randomness function. This corresponds to the case where Eve measures each purification system, $\rho_E=\tr_A(\ketbra{\Psi}^{AE})$, with an optimal i.i.d.~measurement. Since Eve can optimize her measurement, in randomness evaluation, one should minimize over all possible decomposition of $\rho$ \cite{yuan2015intrinsic}. For a decomposition $\rho=\sum q_j\ketbra{\psi_j}$, where $\sum_j q_j=1$ and $q_j\geq0,\forall j$,
\begin{equation} \label{appeq:1PVMRc}
\begin{split}	R_c(\rho,\{\ketbra{i}\})&=\underset{\{q_j,\ket{\psi_j}\}}{\min}\sum_j q_j S(\Delta_{\{\ketbra{i}\}}(\ketbra{\psi_j})). \\
\end{split}
\end{equation}
In the second case,
Eve can collect all the purification states and perform a joint measurement. In this case, randomness is quantified by \cite{Devetak2005distillation,yuan2019quantum},
\begin{equation}\label{appeq:RqvonNeum}
\begin{split}		
R_q(\rho,\{\ketbra{i}\})&=S(\rho\|\Delta_{\{\ketbra{i}\}}(\rho)),
\end{split}
\end{equation}
where $S(\rho\| \sigma)=\tr(\rho\log\rho-\rho\log\sigma)$ is the quantum relative entropy function.


The intrinsic randomness functions, Eq.~\eqref{appeq:1PVMRc} and \eqref{appeq:RqvonNeum} directly give coherence measures for both $R_c$ \cite{yuan2015intrinsic} and $R_q$ \cite{yuan2019quantum,Hayashi2018secure} for von Neumann measurements.

\subsection{Block randomness}
\begin{definition}[Dephasing operation \cite{Aberg2006}]
	For a projection measurement, $\mathbf{P}=\{ P_{1},\cdots,P_{m} \}$, acting on a $d$-dimensional Hilbert space $ \mathcal{H}$, the corresponding block-dephasing operation is defined as
\begin{equation}
\Delta_{\mathbf{P}}(\rho)=\sum_{i=1}^{m}P_{i}\rho P_{i}.
\end{equation}
\end{definition}

Consider a general PVM, when the input state is pure, the randomness of outcomes is again directly given by Born's rule,
\begin{equation}\label{appeq:Rpurepvm}
	\begin{split}
R_c(\ketbra{\psi},\mathbf{P}) = S(\Delta_{\mathbf{P}}(\ketbra{\psi})).
	\end{split}
\end{equation}
For a general input state $\rho$, similar to the case of a von Neumann measurement, we have two ways to extend the results from the pure state case. First, if Eve measures the purification system, $\rho_E$, on an optimal basis for her, the randomness of the PVM outcomes is quantified by the convex-roof measure,
\begin{equation}\label{appeq:Rcpvm}
\begin{split} R_c(\rho,\mathbf{P})&=\underset{\{q_j,\ket{\psi_j}\}}{\min}\sum_j q_j S(\Delta_{\mathbf{P}}(\ketbra{\psi_j})). \\
\end{split}
\end{equation}
Second, if Eve performs a joint measurement on the independent purification systems, the intrinsic randomness is characterized by the von Neumann entropy of the post-measurement state of Alice conditioned on Eve. Alice obtains the classical outcomes and stores them in system $A'$, which is $m$-dimensional. After Alice's measurement, the classical-quantum state is given by,
\begin{equation}\label{appeq:postPVMstate}
\tilde{\rho}^{A'E}=\sum_i p_i\ketbra{i}^{A'} \otimes \rho^E_i,
\end{equation}
where $p_i=\tr(P_i\rho)$ is the probability obtaining output $i$ and $\rho^E_i=\frac{1}{p_i}\tr_A\left(\ketbra{\Psi}^{AE}(P_i^A\otimes \mathbf{1}^E)\right)$. As Alice's measurement should not change the state of Eve's system, we have
\begin{equation}
\tr_A(\ketbra{\Psi}^{AE})=\rho^{E}=\tr_A(\tilde{\rho}^{A'E})=\sum_ip_i \rho^E_i.
\end{equation}
The intrinsic randomness of Alice's measurement result, defined by the quantum entropy conditioned on Eve,
\begin{equation}\label{appeq:Rqpvm}
\begin{split}
R_q(\rho,\mathbf{P})&=S(A'|E)_{\tilde{\rho}^{A'E}} \\
&=S(\tilde{\rho}^{A'E})-S(\rho^{E})\\
&=H(\{p_i\})+\sum_i p_i S(\rho^E_i)-S(\rho)\\
&=H(\{p_i\})+\sum_i p_i S\left(\frac{P_i\rho P_i}{p_i}\right)-S(\rho)\\
&=S\left(\sum_i P_i\rho P_i\right)-S(\rho)\\
&=S(\rho\|\Delta_{\mathbf{P}}(\rho)),
\end{split}
\end{equation}
where $H(\{p_i\})=-\sum_i p_i\log p_i$ is the Shannon entropy function, and $\rho$ denotes Alice's initial state.
The fourth equality is due to the fact that
$\rho^E_i$ is a pure state. The result is consistent with Eq.~\eqref{appeq:RqvonNeum}. Similar to the case of a von Neumann measurement, the difference between two randomness function is quantified by the discord between $A'$ and $E$ \cite{yuan2019quantum},
\begin{equation}\label{appeq:diffdisc} R_c(\rho,\mathbf{P})-R_q(\rho,\mathbf{P})=D_E(\tilde{\rho}^{A'E})=\underset{\{q_j^E\}}{\min} S(A'|\{q_j^E\})_{\tilde{\rho}^{A'E}}-S(\tilde{\rho}^{A'E})+S(\rho^{E}),
	\end{equation}
where $\{q_j^E\}$ is a probability distribution given by measurement on system $ E $. The equality holds since the measurement on $E$ corresponds to a decomposition $\rho^A=\sum_j q_j^E \ketbra{\psi_j}$ and $S(A'|\{q_j^E\})_{\tilde{\rho}^{A'E}}=\sum_j q_j^E S(\Delta_{\mathbf{P}}(\ketbra{\psi_j}))$.

It is straightforward to check that both randomness functions, Eq.~\eqref{appeq:Rcpvm} and \eqref{appeq:Rqpvm}, satisfy the convexity condition,
\begin{equation}\label{appeq:RPconvexApp}
\begin{split}
		R_c\left(\sum_j r_j \rho_j, \mathbf{P}\right) &\le \sum_j r_j R_c(\rho_j, \mathbf{P}),  \\
		R_q\left(\sum_j r_j \rho_j, \mathbf{P}\right) &\le \sum_j r_j R_q(\rho_j, \mathbf{P}).
\end{split}
\end{equation}
As for the additivity condition for block-diagonal states, it is less obvious.

\begin{definition}[Block-diagonal state with respect to $\mathbf{P}$]\label{appdef:blockdiagP}
A state is called block-diagonal with respect to $\mathbf{P}$, denoted by $\rho=\sum_j r_j \rho_j \equiv\oplus_jr_j \rho_j$ with $\sum r_j=1$ and $r_j\ge0,\forall j$, if $\forall P_i, j\neq k$,
\begin{equation}\label{appeq:blockdiagP}
\tr(P_i \rho_j P_i \rho_k)=0.
\end{equation}
\end{definition}

For a block-diagonal state, each PVM element projects diagonal decomposed state $\rho_j$ into density operators with different orthogonal supports. From the definition, we can show that the diagonal decomposed states are orthogonal, $\forall j\neq k$, $\tr(\rho_j \rho_k)=0$. Consider the spectral decomposition of $\rho_j=\sum_a \lambda_a \ketbra{\alpha_a}$ and $\rho_k=\sum_b \mu_b \ketbra{\beta_b}$, where $\lambda_a>0$ and $\mu_b>0$, by Eq.~\eqref{appeq:blockdiagP},
\begin{equation}\label{appeq:orthogonal}
\begin{split}
\tr(P_i \rho_j P_i \rho_k) &= \tr[P_i \left(\sum_a \lambda_a \ketbra{\alpha_a}\right) P_i \left(\sum_b \mu_b \ketbra{\beta_b}\right)] \\
&=\sum_{a,b} \lambda_a \mu_b \tr(P_i \ketbra{\alpha_a} P_i \ketbra{\beta_b}) \\
&=\sum_{a,b} \lambda_a \mu_b \abs{\bra{\alpha_a} P_i \ket{\beta_b}}^2 \\
&=0. \\
\end{split}
\end{equation}
Hence, $\forall a,b$, $\bra{\alpha_a} P_i \ket{\beta_b}=0$, and $\forall P_i$,
\begin{equation}\label{appeq:orthogonalPis}
	\begin{split}
\sum_{a,b} \ketbra{\alpha_a} P_i \ketbra{\beta_b} =0.
	\end{split}
\end{equation}
Moreover, there is
\begin{equation}
	\begin{split}
\tr(\rho_j \rho_k) &= \tr[\rho_j \left(\sum_i P_i\right) \rho_k] \\
&= \sum_i \tr(\rho_j P_i \rho_k) \\
&= \sum_i \sum_{a,b} \lambda_a \mu_b \bra{\alpha_a} P_i \ket{\beta_b} \braket{\beta_b}{\alpha_a} \\
&=0.
	\end{split}
\end{equation}

\begin{lemma}[Additivity of randomness functions]\label{applem:addblockdiag}
For a block-diagonal state $\rho=\oplus_j r_j \rho_j$ with respective to $\mathbf{P}$, the randomness functions satisfy the additivity condition,
\begin{equation}
\begin{split}
R_c(\oplus_j r_j \rho_j,\mathbf{P})&=\sum_j r_jR_c(\rho_j,\mathbf{P}), \\
R_q(\oplus_j r_j \rho_j,\mathbf{P})&=\sum_j r_jR_q(\rho_j,\mathbf{P}).
\end{split}
\end{equation}
\end{lemma}
\begin{proof}
Without loss of generality, we consider the two-block case, $\rho=r \rho_1\oplus(1-r)\rho_2$.

Let us first prove for the case of $R_c$. Denote $\Pi_1$ as the projector onto the support of $\rho_1$, $\Pi_1=\sum_a \ketbra{\alpha_a}$ with $\{\ket{\alpha_a}\}$ being the set of non-zero eigenstates of $\rho_1$. Similarly, denote $\Pi_2$ as the projector onto the support of state $\rho_2$. In this proof, only the support of $\rho$ is under consideration, $\Pi_1+\Pi_2=\mathbf{1}$. Then, $\Pi_1\rho\Pi_1=r\rho_1$ and $\Pi_2\rho\Pi_2=(1-r)\rho_2$. 


For arbitrary pure state $\ket{\phi}$ in the support space of $\rho$, let $ \ket{\phi_1}=\Pi_1\ket{\phi}/\sqrt{p_1} $ and
$\ket{\phi_2}=\Pi_2\ket{\phi}/\sqrt{p_2}$ with $ p_1= \bra{\phi}\Pi_1\ket{\phi}$ and
$ p_2=1-p_1= \bra{\phi}\Pi_2\ket{\phi}$, and from Eq.~\eqref{appeq:orthogonalPis},
\begin{equation}
\bra{\phi}\Pi_1 P_i \Pi_1\ket{\phi}+\bra{\phi}\Pi_2 P_i \Pi_2\ket{\phi}=\bra{\phi} P_i \ket{\phi}.
\end{equation}
The quantum entropies,
\begin{equation}\label{appeq:SPphi12}
\begin{split} S(\Delta_{\mathbf{P}}(\ketbra{\phi_1}))&=H\left(\left\{\frac{\bra{\phi}\Pi_1 P_i \Pi_1\ket{\phi}}{p_1}\right\}\right), \\ S(\Delta_{\mathbf{P}}(\ketbra{\phi_2}))&=H\left(\left\{\frac{\bra{\phi}\Pi_2 P_i \Pi_2\ket{\phi}}{p_2}\right\}\right),
\end{split}
\end{equation}
satisfy the concavity condition,
\begin{equation}\label{appeq:Rcconcavity}
p_1 S(\Delta_{\mathbf{P}}(\ketbra{\phi_1}))+p_2 S(\Delta_{\mathbf{P}}(\ketbra{\phi_2}))\le H(\{\bra{\phi}P_i \ket{\phi}\})= S(\Delta_{\mathbf{P}}(\ketbra{\phi})).
\end{equation}

Now, consider an arbitrary decomposition $\rho=\sum_j q_j \ketbra{\psi_j}$,
\begin{equation}
	\begin{split}
\Pi_1\rho\Pi_1 &=r\rho_1= \sum_j {q_j p_{1j}} \ketbra{\psi_{1j}}, \\
\Pi_2\rho\Pi_2 &=(1-r)\rho_2= \sum_j {q_j p_{2j}} \ketbra{\psi_{2j}}, \\
	\end{split}
\end{equation}
where $p_{1j}= \bra{\psi_j}\Pi_1\ket{\psi_j}$, $p_{2j}= \bra{\psi_j}\Pi_2\ket{\psi_j}$, $\ket{\psi_{1j}}= \Pi_1\ket{\psi_j}/\sqrt{p_{1j}}$, and $\ket{\psi_{2j}}= \Pi_2\ket{\psi_j}/\sqrt{p_{2j}}$. Combining with Eq.~\eqref{appeq:Rcconcavity},
\begin{equation}
\begin{split}
\sum_j q_j p_{1j} S(\Delta_{\mathbf{P}}(\ketbra{\psi_{1j}}))+\sum_j q_j p_{2j} S(\Delta_{\mathbf{P}}(\ketbra{\psi_{2j}})) \le \sum_j q_j S(\Delta_{\mathbf{P}}(\ketbra{\psi_j})).
\end{split}
\end{equation}
Note that $\{\ket{\psi_{1j}}\}$ and $\{\ket{\psi_{2j}}\}$ are the pure state decompositions of $\rho_1$ and $\rho_2$, respectively. According to the randomness function defined in Eq.~\eqref{appeq:Rcpvm}, we have $rR_c(\rho_1,\mathbf{P})+(1-r)R_c(\rho_2,\mathbf{P})\le R_c(\rho,\mathbf{P})$. Therefore, by combining it with the convexity condition, Eq.~\eqref{appeq:RPconvexApp}, the additivity condition can be obtained.

For the case of $R_q$, first, from the definition in Eq.~\eqref{appeq:blockdiagP}, $\rho_1$ and $\rho_2$ have different orthogonal supports,
\begin{equation}\label{appeq:Srho}
	\begin{split}
S(\rho) &=S(r \rho_1+(1-r)\rho_2) \\
&=h(r)+rS(\rho_1)+(1-r)S(\rho_2), \\
	\end{split}
\end{equation}
where $h(r)=-r\log r-(1-r)\log(1-r)$ is the binary entropy function. Second, the dephasing state is given by,
\begin{equation}
	\begin{split}
\sum_i P_i\rho P_i 
&= \sum_i r P_i\rho_1P_i + (1-r)P_i\rho_2 P_i \\
&= r\Delta_{\mathbf{P}}(\rho_1) + (1-r)\Delta_{\mathbf{P}}(\rho_2). \\
	\end{split}
\end{equation}
From the definition in Eq.~\eqref{appeq:blockdiagP}, we can see that $\Delta_{\mathbf{P}}(\rho_1)$ and $\Delta_{\mathbf{P}}(\rho_2)$ have orthogonal supports, and hence,
\begin{equation}\label{appeq:Sdephasingrho}
\begin{split}
S\left(\sum_i P_i\rho P_i\right) &= h(r) + r S(\Delta_{\mathbf{P}}(\rho_1)) + (1-r)S(\Delta_{\mathbf{P}}(\rho_2)). \\
\end{split}
\end{equation}
By substituting Eq.~\eqref{appeq:Srho} and \eqref{appeq:Sdephasingrho} into the fifth equality of Eq.~\eqref{appeq:Rqpvm}, we can prove the claim.


\end{proof}

\begin{lemma}[Randomness-invariant unitary]\label{applem:Purdecran}
	For two PVMs $ \mathbf{P}=\{P_1,\cdots,P_m\} $ and $U^{\dagger} \mathbf{P}U=\{U^{\dagger} P_1U,\cdots,U^{\dagger}P_mU\}$ connected by a unitary operator $U$, if for any pure state $\ket{\psi}$ in the support of an input state $\rho$, $\forall i$,
	\begin{equation}\label{appeq:rhoUconstraint}
		\bra{\psi}P_i\ket{\psi}=\bra{\psi}U^{\dagger}P_iU\ket{\psi},
	\end{equation}
	then the randomness of $\rho$ with respect to these two PVMs are the same,
	\begin{equation}
		R(\rho,\mathbf{P})=R(\rho,U^{\dagger} \mathbf{P}U).
	\end{equation}
\end{lemma}
\begin{proof}
	According to the unitary invariance of the randomness functions, we always have
	$R(\rho,U^{\dagger} \mathbf{P}U)=R(U\rho U^{\dagger} ,\mathbf{P})$ and need to prove $ R(\rho,\mathbf{P})=R(U\rho U^{\dagger} ,\mathbf{P})$.
	
(i) Consider $R=R_c$,
	\begin{equation}\label{appeq:RcPVMUsame}
		\begin{split}
			R_c(\rho,\mathbf{P})&=\underset{\{q_j,\ket{\psi_j}\}}{\min}\sum_j q_j R_c(\ketbra{\psi_j},\mathbf{P})\\
			&=\underset{\{q_j,\ket{\psi_j}\}}{\min}\sum_j q_j R_c(U\ketbra{\psi_j}U^{\dagger},\mathbf{P})\\
			&=R_c(U\rho U^{\dagger} ,\mathbf{P}).
		\end{split}
	\end{equation}
	
(ii) Consider the case $R=R_q$ and suppose $ \rho=\sum_j \lambda_j \ketbra{\psi_j} $ is the spectral decomposition, its purification can be taken as $ \ket{\Psi^{AE}}=\sum_j \sqrt{\lambda}_j \ket{\psi_j}^A \ket{j}^E$. Denote $ \sigma= U\rho U^{\dagger}$ and $(U^A \otimes \mathbf{1}^E) \ket{\Psi^{AE}}$ is its purification. Write the two classical-quantum states after Alice's measurement as given in Eq.~\eqref{appeq:postPVMstate},
\begin{equation}\label{appeq:cqPVMU}
\begin{split}
\tilde{\rho}^{A'E}&=\sum_i p_i\ketbra{i}^{A'} \otimes \rho^E_i,\\
\tilde{\sigma}^{A'E}&=\sum_i p_i\ketbra{i}^{A'} \otimes \sigma^E_i,
\end{split}
\end{equation}
where from Eq.~\eqref{appeq:rhoUconstraint}, $p_i$ and the classical measurement outcome system $A'$ are the same in the two equations, $\rho^E_i=\tr_A[\ketbra{\Psi}^{AE}(P_i^A\otimes \mathbf{1}^E)] /p_i$ and $\sigma^E_i=\tr_A[\ketbra{\Psi}^{AE}((U^{\dagger} P_i U)^A\otimes \mathbf{1}^E)] /p_i$.

For any two different eigenstates $\ket{\psi_j}$ and $\ket{\psi_{j'}}$, consider a pure state $ \ket{\psi}=a \ket{\psi_j}+b\ket{\psi_{j'}}$, with $ \abs{a}^2+ \abs{b}^2=1$. According to Eq.~\eqref{appeq:rhoUconstraint},
\begin{equation}
\abs{a}^2 p_{ij}+\abs{b}^2 p_{ij'}+2\Re(a^*b\bra{\psi_{j}}P_i\ket{\psi_{j'}})=\abs{a}^2 p_{ij}+\abs{b}^2 p_{ij'}+2\Re(a^*b\bra{\psi_{j}}U^{\dagger}P_i U\ket{\psi_{j'}}),
\end{equation}
where $p_{ij}=\bra{\psi_{j}}P_i\ket{\psi_{j}}$ and $p_{ij'}=\bra{\psi_{j'}}P_i\ket{\psi_{j'}}$. Since $ a,b $ are arbitrary, it can be obtained $\bra{\psi_{j}}P_i\ket{\psi_{j'}}=\bra{\psi_{j}}U^{\dagger}P_i U\ket{\psi_{j'}} ,\forall j,j'$. Combined with the fact
\begin{equation}
\begin{split}
\rho^E_i&=\frac{1}{p_i}\sum_{jj'}\sqrt{\lambda_j\lambda_{j'}}\bra{\psi_{j'}}P_i\ket{\psi_{j}}\ketbra{j}{j'}^E, \\
\sigma^E_i&=\frac{1}{p_i}\sum_{jj'}\sqrt{\lambda_j\lambda_{j'}}\bra{\psi_{j'}}U^{\dagger}P_i U\ket{\psi_{j}}\ketbra{j}{j'}^E,
\end{split}
\end{equation}
we have, $\rho^{E}=\sigma^{E}$. Then, from Eq.~\eqref{appeq:cqPVMU}, we get $\tilde{\rho}^{A'E}=\tilde{\sigma}^{A'E}$. Moreover, due to Eq.~\eqref{appeq:RcPVMUsame} and the relation between two randomness functions Eq.~\eqref{appeq:diffdisc}, we have $R_q(\rho,\mathbf{P})=R_q(\rho,U^{\dagger} \mathbf{P}U)$.
\end{proof}

\subsection{Block coherence}
The randomness functions developed in Eq.~\eqref{appeq:Rcpvm} and \eqref{appeq:Rqpvm} can also be used for the block coherence measures, which was introduced by {\AA}berg \cite{Aberg2006} and was put in a resource theory framework later \cite{Bischof2019}. Here, we adopt a slightly reformulated framework. For a PVM $\mathbf{P}$, the set of block-incoherent states is defined as
\begin{equation}\label{appeq:incstateset}
\mathcal{I}_{\mathbf{P}}=\left\{\Delta_{\mathbf{P}}(\rho)\mid\rho \in  \mathcal{D(H)} \right\},
\end{equation}
with $\mathcal{D(H)}$ being the set of all the density matrices defined on the Hilbert space.

There are different ways to define free operations. For example, we can consider block-incoherent operations (IO),
\begin{equation}\label{appeq:IO}
\Lambda_{IO}(\rho)=\sum_n K_n\rho K_n^{\dagger},
\end{equation}
with Kraus operators satisfying
\begin{equation}
K_n \mathcal{I}_{\mathbf{P}} K_n^{\dagger}\subseteq \mathcal{I}_{\mathbf{P}}.
\end{equation}

A block coherence measure should satisfy the criteria presented in Box \ref{appbox:PCoherence}, which are essentially the same as the coherence-measure criteria for von Neumann measurements \cite{streltsov2017colloquium}. These criteria are not independent. In fact, criteria (C3) and (C4) for incoherent operations defined in Eq.~\eqref{appeq:IO} can be replaced by the additivity for the block-diagonal state defined in Definition \ref{appdef:blockdiagP}, $C(\oplus_j \lambda_j \rho_j)=\sum_j \lambda_jC(\rho_j)$ \cite{Yu2016Alternative,streltsov2017colloquium}. Here, we omit PVM $\mathbf{P}$ from the coherence measure $C(\rho,\mathbf{P})$ for simplicity when there is no confusion.

\begin{mybox}[label={appbox:PCoherence}]{{Criteria for block-coherence measures for PVMs}}
\begin{enumerate}[(C1)]
	\item
	Nonnegativity: $C(\rho)\geq 0 ,$ and $ C(\sigma)=0 $ iff $ \sigma\in\mathcal{I}_{\mathbf{P}}$;
	
	\item
	Monotonicity: for any incoherent map $\Lambda_{I}$, $C(\rho)\ge C(\Lambda_{I}(\rho))$;
	
	\item
	Strong monotonicity: for any IO map defined in Eq.~\eqref{appeq:IO}, $C(\rho)\ge \sum_n p_n C(\rho_n) $, where $p_n=\tr(K_n\rho K_n^{\dagger}) $ and $ \rho_n=K_n\rho K_n^{\dagger}/p_n $;
	
	\item
	Convexity: $C(\sum_{j}\lambda_{j}\rho_{j})\leq \sum_{j}\lambda_j C(\rho_{j})$ with $\lambda_{j}>0$ and $\sum \lambda_{j}=1$;
	
	\item
	Uniqueness for pure states: $C(\ketbra{\psi}) = S(\Delta_{\mathbf{P}}(\ketbra{\psi}))$;

\item
Additivity for a tensor state, $C(\rho\otimes \delta,\mathbf{P}_1\otimes\mathbf{P}_2)= C(\rho,\mathbf{P}_1)+C(\delta,\mathbf{P}_2)$, where the joint projection measurement on $\rho\otimes \delta$ takes the tensor form of local projectors on $\rho$ and $\delta$.
	
\end{enumerate}
\end{mybox}

With criterion (C5), we can employ the convex-roof construction for a block coherence measure. Not surprisingly, same as the von Neumann measurement case, the measure is given by the randomness function $R_c$ defined in Eq.~\eqref{appeq:Rcpvm}.

We can also define block coherence measure $C(\rho)$ via the relative entropy, which, again, is the same as randomness function Eq.~\eqref{appeq:Rqpvm},
\begin{equation}
\begin{split}
C(\rho)&=\underset{\sigma\in\mathcal{I}_{\mathbf{P}}}{\min}S(\rho\|\sigma)\\
&=-\underset{\sigma\in\mathcal{I}_{\mathbf{P}}}{\max}\tr(\rho \log \sigma)-S(\rho)\\
&=-\underset{\sigma\in\mathcal{I}_{\mathbf{P}}}{\max}\tr(\Delta_{\mathbf{P}}(\rho) \log \sigma)-S(\rho)\\
&=S(\Delta_{\mathbf{P}}(\rho))-S(\rho)+\underset{\sigma\in\mathcal{I}_{\mathbf{P}}}{\min}S(\Delta_{\mathbf{P}}(\rho)\|\sigma)\\
&=S(\Delta_{\mathbf{P}}(\rho))-S(\rho)\\
&=R_q(\rho,\mathbf{P}),
\end{split}
\end{equation}
where the third equality holds because any incoherent state $\sigma$ in Eq.~\eqref{appeq:incstateset} is diagonal with respective to $\mathbf{P}$.

\section{Generalized Naimark extension}\label{App:Generalizedextension}

\begin{lemma}[Consistency condition]
A Naimark extension should satisfy the consistency condition, $\forall \rho^A$, $\tr(M_i\rho^A)=\tr[P_i(\rho^A\otimes \sigma^Q)]$, which is equivalent to,
	\begin{equation}\label{appeq:consisMitrQ}
		M_i= \tr_{Q}[P_i(\mathbf{1}^{A}\otimes\sigma^{Q})].
	\end{equation}
\end{lemma}

\begin{proof}
$\forall \rho^A$,
	\begin{equation}
		\begin{split}
			\tr(M_i \rho^A)&=\tr(P_i(\mathbf{1}^{A}\otimes\sigma^{Q})( \rho^A \otimes \mathbf{1}^{Q} ))\\
			&=\tr_{A}\left(\rho^{A}\tr_{Q}[P_i(\mathbf{1}^{A}\otimes\sigma^{Q})]\right),
		\end{split}
	\end{equation}
which concludes Eq.~\eqref{appeq:consisMitrQ}.
\end{proof}

Consider a Naimark extension for a POVM, as shown in Figure~\ref{appfig:pvmeq}(a). One can also view the extended PVM as a unitary followed by a von Neumann measurement, also called rank-1 PVM --- a set of rank-1 projectors, as shown in Figure~\ref{appfig:pvmeq}(b). Now, the extended PVM $\mathbf{P}=\{P_1,\cdots,P_{m}\}  $ can be written as
\begin{equation}\label{appeq:PiextU}
	P_i=U^{\dagger}(\mathbf{1}^{A'}\otimes\ketbra{i})U,
\end{equation}
where system $A'$ is generally different from $A$. The unitary operation is sometimes called preprocessing \cite{ozawa1984quantum,buscemi2008global}.

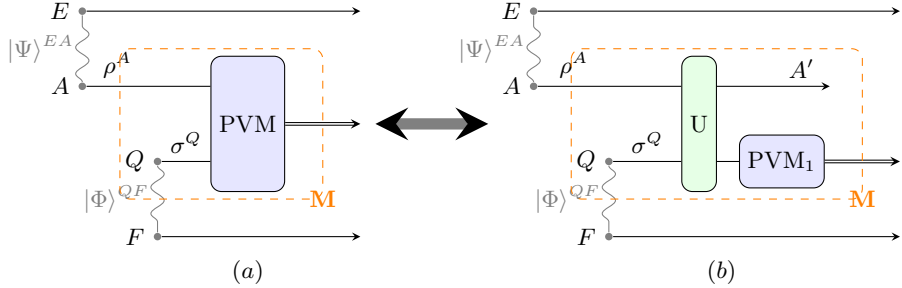
\begin{figure}[hbtp!]
	\begin{tikzpicture}[
		scale=1,
		]
		\node[circle,fill,gray,inner sep=1pt,label=left:$A$] (A) at (0,0) {};
		\node[circle,fill,gray,inner sep=1pt,label=left:$Q$] (Q) at ($ (A) + (1,-1) $) {};
		\node[circle,fill,gray,inner sep=1pt,label=left:$E$] (E) at ($ (A) + (0,1) $) {};
		\node[circle,fill,gray,inner sep=1pt,label=left:$F$] (F) at ($ (A) + (1,-2) $) {};
		\draw[snake=snake,gray] (A) -- (E) node[midway,left] {$\ket{\Psi}^{EA}$};
		\draw[snake=snake,gray] (Q) -- (F) node[midway,left] {$\ket{\Phi}^{QF}$};
		\node[draw,black,fill=blue!10!white,minimum height=1.8cm,rounded corners] (M) at ($ (A) + (2.2,-.5) $) {PVM};
		\draw (A) -- (A-|M.west) node[near start,above] {$\rho^A$} (Q) -- (Q-|M.west) node[midway,above] {$\sigma^Q$};
		\draw[-stealth,double] (M.east) -- ++ (1cm,0) node[near start,above] {};
		\draw[-stealth] (E) -- ($(E-|M.east) + (1cm, 0) $);
		\draw[-stealth] (F) -- ($(F-|M.east) + (1cm, 0) $);
		\draw[orange,dashed,rounded corners] ($ (A-|Q)!0.5!(E) $) rectangle ($(Q-|M.east)!0.5!(F-|M.east)+(.5cm,0)$) node{$\mathbf{M}$};
		
		\coordinate[label = below:$(a)$] (a) at (2.2, -2.2);
		
		\draw[stealth-stealth,line width=1.5mm,draw=gray] ($(M.east) + (1.2cm, 0) $) -- ++ (1.5cm,0);
		
		\node[circle,fill,gray,inner sep=1pt,label=left:$A$] (Ab) at (6cm,0) {};
		\node[circle,fill,gray,inner sep=1pt,label=left:$Q$] (Qb) at ($ (Ab) + (1,-1) $) {};
		\node[circle,fill,gray,inner sep=1pt,label=left:$E$] (Eb) at ($ (Ab) + (0,1) $) {};
		\node[circle,fill,gray,inner sep=1pt,label=left:$F$] (Fb) at ($ (Ab) + (1,-2) $) {};
		\draw[snake=snake,gray] (Ab) -- (Eb) node[midway,left] {$\ket{\Psi}^{EA}$};
		\draw[snake=snake,gray] (Qb) -- (Fb) node[midway,left] {$\ket{\Phi}^{QF}$};
		\node[draw,black,fill=green!10!white,minimum height=1.8cm,rounded corners] (U) at ($ (Ab) + (2.2,-.5) $) {U};
		\node[draw,black,fill=blue!10!white,minimum size=20pt,rounded corners] (Mb) at ($ (Qb) + (2.3,0) $) {PVM$_1$};
		\draw (Ab) -- (Ab-|U.west) node[near start,above] {$\rho^A$} (Qb) -- (Qb-|U.west) node[midway,above] {$\sigma^Q$} (Qb-|U.east) -- (Mb.west);
		\draw[-stealth] (Ab-|U.east) -- ++ (1.5cm,0) node[near end,above] {$A'$};
		\draw[-stealth,double] (Mb.east) -- ++ (1cm,0) node[near start,above] {};
		\draw[-stealth] (Eb) -- ($(Eb-|Mb.east) + (1cm, 0) $);
		\draw[-stealth] (Fb) -- ($(Fb-|Mb.east) + (1cm, 0) $);
		\draw[orange,dashed,rounded corners] ($ (Ab-|Qb)!0.5!(Eb) $) rectangle ($(Qb-|Mb.east)!0.5!(Fb-|Mb.east)+(.5cm,0)$) node{$\mathbf{M}$};
		\coordinate[label = below:$(b)$] (b) at ($ (Ab) + (2.5,-2.2) $) {};
	\end{tikzpicture}
	\caption{Two equivalent pictures of the adversary scenario. (a) Naimark extension: a PVM on the joint system $AQ$. (b) Equivalent extension: a joint preprocessing unitary operation on system $AQ$ followed by a rank-1 PVM on some subsystem.} \label{appfig:pvmeq}
\end{figure}

\begin{lemma}[Equivalence between two forms of Naimark extensions]\label{applem:UPOVMrepre}
For any generalized Naimark extension of a POVM, $\{\mathbf{\tilde{P}},\sigma^Q\}$, we can find an equivalent extension, $\{\mathbf{P},\sigma^Q\}$, in the form of Eq.~\eqref{appeq:PiextU}, as shown in Figure~\ref{appfig:pvmeq}. That is, $\forall i, \rho^A$,
\begin{equation}
	\begin{split}
P_i(\rho^{A}\otimes\sigma^{Q})P_i &= \tilde{P}_i(\rho^{A}\otimes\sigma^{Q})\tilde{P}_i,
	\end{split}
\end{equation}
where we ignore a trivial zero subspace.
\end{lemma}

From Eq.~\eqref{appeq:PiextU}, we can see that all the PVM elements, $P_i$, in this extension have the same ranks. The key point to prove the equivalence between the two forms of extensions, Figure~\ref{appfig:pvmeq}(a) and (b), is to show that any extension is equivalent to an extension with same rank PVM elements.

\begin{proof}
For an extended PVM $\mathbf{\tilde{P}}=\{\tilde{P}_1,\cdots,\tilde{P}_{m}\}$, denote the maximum rank of the elements as $ r,$ and assume $ s $ is an integer satisfying $ (s-1)d<r\leq sd$. Consider a larger space $\mathcal{H}^{\prime}$ with dimension $msd$ and embed $\mathcal{H}_{AQ}$ into $\mathcal{H}^{\prime}$. For each ${\rm rank}(\tilde{P}_i)<sd$, we can further extend $\tilde{P}_i$ to rank-$sd$ by adding additional rank-1 projectors $\{\ketbra{\varphi}\}$, where $\{\ket{\varphi}\}$ are normalized basis states chosen from the complement space of $\mathcal{H}_{AQ}$, $\mathcal{H}^{\prime}\ominus \mathcal{H}_{AQ}$. In the end, we can have a set of $m$ new extended PVM elements, $\mathbf{P}=\{{P}_1,\cdots,{P}_{m}\}$, whose ranks are $sd$ in $\mathcal{H}^{\prime}$ and $dim(\mathcal{H}^{\prime})=msd$. Note that these newly added complement projectors, $\{\ketbra{\varphi}\}$, are orthogonal to the state space $\mathcal{D}(\mathcal{H}_{A})\otimes \sigma^Q$, so, they have no influence on the measurement outcomes. Now, all the elements have the same ranks, the extended PVM is unitarily equivalent to $\{\ketbra{1}\otimes\mathbf{1},\cdots,\ketbra{m}\otimes \mathbf{1}\}$, i.e., $P_i=U^{\dagger} (\ketbra{i}\otimes \mathbf{1})U$, as given in Eq.~\eqref{appeq:PiextU}. 
\end{proof}

Now, we can introduce canonical Naimark extension, where systems $A$ and $A'$ are the same in Figure~\ref{appfig:pvmeq}(b) and the dimension of the ancillary system is the same as the number of POVM elements.

\begin{definition}[Canonical Naimark extension \cite{peres1990neumark}]\label{appdef:CanExt}
For POVM $\mathbf{M}=\{M_{1},\cdots,M_{m}\}$ in $\mathcal{H}_A$,  the canonical Naimark extension results in a PVM $\mathbf{P} =\{P_{1},\cdots,P_{m}\}$ in a larger Hilbert space  $\mathcal{H}_A\otimes\mathcal{H}_{Q},$ where $ \mathcal{H}_{Q} $ is an $m$-dimensional ancillary space. Assume $\{\ket{1},\cdots,\ket{m}\}$ is an orthonormal basis of ancillary space $\mathcal{H}_{Q}$. The ancillary state is $\ketbra{1}$. Then, each $P_i$ has the form of
\begin{equation}\label{appeq:ne}
	P_i=U^{\dagger}(\mathbf{1}^{A} \otimes \ketbra{i}^{Q})U,
\end{equation}
with $U$ being a suitable unitary operator to satisfy the consistency condition
\begin{equation}
M_i=\tr_Q[P_i( \mathbf{1}^A \otimes \ketbra{1}^Q)].
\end{equation}
\end{definition}

\begin{proposition}[Correspondence between ancillary state and measurement decomposition]\label{appPro:CorrDecomp}
In a generalized Naimark extension of a POVM, $\mathbf{M}$, if the ancillary state has a pure state decomposition, $\sigma=\sum_j r_j \ketbra{\varphi_j}$ with $\sum r_j=1$ and $r_j>0$, then there exists a measurement decomposition  $\mathbf{M}=\sum_{j} r_j \mathbf{N}^j$, and vice versa.
\end{proposition}

\begin{proof}
Denote a generalized Naimark extension to be $\{\mathbf{P},\sigma\}$.

$\Rightarrow$: The ancillary state has decomposition of $\sigma=\sum_j r_j \ketbra{\varphi_j}$. Then, the consistency condition becomes,
\begin{equation}
\begin{aligned}
M_i &= \tr_{Q}[P_i(\mathbf{1}^{A}\otimes\sigma^{Q})] \\
&=\sum_{j} r_j \tr_{Q}[P_i(\mathbf{1}^{A}\otimes \ketbra{\varphi_j}^Q)]. \\
\end{aligned}
\end{equation}
Denote $N_i^j=\tr_{Q}[P_i(\mathbf{1}^{A}\otimes \ketbra{\varphi_j}^Q)]$ and hence,
\begin{equation}
\begin{aligned}
M_i = \sum_{j} r_jN_i^j.
\end{aligned}
\end{equation}
Here, $\mathbf{N}^j=\{N_1^j,\cdots,N_{m}^j\} $ forms a POVM on system $ A $ for each $j$. Thus, this gives a decomposition of POVM $ \mathbf{M}=\sum_{j} r_j \mathbf{N}^j$. Note that each $\mathbf{N}^j$ has an extended PVM of $\mathbf{P}$, independent of $j$.

$\Leftarrow$: The POVM has decomposition of $\mathbf{M}=\sum_{j=1}^l r_j \mathbf{N}^j$. According to the canonical Naimark extension, for each POVM $\mathbf{N}^j$, there exists an ancillary system $ Q_1 $ with an orthonormal basis $ \{|1\rangle,\cdots,|m\rangle\} $ and a unitary operator $U_j^{AQ_1}$ on global system $\mathcal{H}^{A}\otimes\mathcal{H}^{Q_1}$ such that
\begin{equation}
N^j_i=\tr_{Q_1}\left[(U_j^{AQ_1})^{\dagger}(\mathbf{1}^{A}\otimes \ketbra{i}^{Q_1})U_j^{AQ_1}(\mathbf{1}^{A}\otimes \ketbra{1} ^{Q_1} )\right].
 \end{equation}
Add another ancillary system $Q_2$ with an orthonormal basis $\{\ket{\varphi_1},\cdots,\ket{\varphi_l}\} $ and $Q_1Q_2$ forms the total ancillary system $Q$ in Figure~\ref{appfig:pvmeq}(b). Let the input ancillary state be $\sigma^Q =\sum_{j=1}^{l}r_j \ketbra{1}^{Q_1}\otimes \ketbra{\varphi_j}^{Q_2}$, which is in a form of pure state decomposition. The unitary operator and the extended PVM are
\begin{equation}\label{appeq:decompN2P}
\begin{split}
U^{AQ} &=\sum_{j=1}^{l}U_j^{A Q_1}\otimes \ketbra{\varphi_j}^{Q_2}, \\
P_i &=(U^{AQ})^{\dagger}(\mathbf{1}^{A Q_2}\otimes \ketbra{i}^{Q_1}) U^{AQ}.
\end{split}
\end{equation}
We need to show the consistency condition,
\begin{equation}
\begin{aligned}
&\tr_Q [P_i(\mathbf{1}^A\otimes \sigma^Q)]\\
&=\tr_Q [ (U^{AQ})^{\dagger}(\mathbf{1}^{A Q_2}\otimes \ketbra{i}^{Q_1}) U^{AQ}(\mathbf{1}^A\otimes \sigma^Q))]\\
&=\sum_{j=1}^{l}  r_j \tr_{Q}(\mathbf{1}^{A}\otimes \ketbra{1\varphi_j} ^{Q} )(U^{AQ})^{\dagger}(\mathbf{1}^{A Q_2}\otimes \ketbra{i}^{Q_1}) U^{AQ}(\mathbf{1}^{A}\otimes \ketbra{1\varphi_j}^{Q} ) \\
&=\sum_{j=1}^{l}  r_j \tr_{Q_1}(\mathbf{1}^{A}\otimes \ketbra{1} ^{Q_1} )(U_j^{A Q_1})^{\dagger}(\mathbf{1}^{A}\otimes \ketbra{i}^{Q_1}) U_j^{A Q_1} (\mathbf{1}^{A}\otimes |1\rangle\langle 1 | ^{Q_1} )\\
&=\sum_{j=1}^{l} r_j N_i^j\\
&=M_i.
\end{aligned}
\end{equation}

\end{proof}

Here, we need to emphasize that the correspondence between the state and measurement decomposition is not unique in general. That is, different $\ketbra{\varphi_j}$ might correspond to the same $\mathbf{N}^j$.

\section{Extremal POVM and its randomness}\label{App:ExtremalPOVM}
The POVM set is convex and some of the POVMs can be treated as a mixture of others,
\begin{equation}
	\mathbf{M}=\sum_{j} r_j \mathbf{N}^j,
\end{equation}
which is equivalent to the existence of a hidden variable in some sense. Those indecomposable POVMs are extreme points of the convex set and play a similar role as pure states \cite{d2005classical,sentis2013decomposition}. Similar to the case for mixed states, general POVMs can be decomposed into a mixture of extremal ones. Here, we introduce the definition and some important properties of extremal POVMs. 

Denote the set of POVMs with at most $ m $ outcomes and the one with discrete outcomes as $ \mathcal{P}(m) $ and $ \mathcal{P}$, respectively. For $m\le n$, $\mathcal{P}(m)\subseteq\mathcal{P}(n)\subseteq\mathcal{P}$, since one can let the $ n-m $ additional elements be 0. The sets $ \mathcal{P}(m)$ and $\mathcal{P}$ are both convex. Clearly, the extremal points in $\mathcal{P}(m)$ are also extremal in $\mathcal{P}$ and the extremal points of $\mathcal{P}$ are called extremal POVMs \cite{chiribella2004extremal,d2005classical}. 

\begin{definition}[Extremal POVM]\label{appdef:ExtremalPOVM}
	A POVM $ \mathbf{M} $ in $ \mathcal{P} $ is said to be extremal if, for every expression 
	\begin{equation}
		\mathbf{M} =\lambda \mathbf{M}^{\prime}+(1-\lambda) \mathbf{M}^{\prime\prime}
	\end{equation} 
	with $ 0<\lambda<1 $ and $ \mathbf{M}^{\prime}, \mathbf{M}^{\prime\prime}\in  \mathcal{P}$, it holds that $ \mathbf{M}=\mathbf{M}^{\prime} =\mathbf{M}^{\prime\prime}$.
\end{definition}
As an example, a PVM is extremal. Next we give a property of extremal POVMs.
\begin{lemma}[Linear independence of extremal POVM elements \cite{d2005classical,Haapasalo2012}]\label{applem:linIndExt}
A POVM is extremal then its elements are linearly independent. A rank-1 POVM is extremal iff its elements are linearly independent.
\end{lemma}
\begin{proof}
Assume elements of the extremal POVM $ \mathbf{M} $ are linearly dependent,
	\begin{equation}
		a_1M_1+\cdots+a_n M_{m}=0,
	\end{equation}
	where not all the coefficients $a_i$'s are 0. Without loss of generality, we can assume $a_1\geq\cdots \geq a_n$. Each element $M_i$ is nonzero positive semidefinite operator, so $a_1>0$, $a_m<0$. Define two new POVMs $\mathbf{M}^{\prime}$ and $\mathbf{M} ^{\prime\prime}$ with, $\forall i\in[m]=\{1,\cdots,m\}$,
	\begin{equation}
		\begin{split}
			M^{\prime}_{i} &=(1-a_i/a_1)M_{i}, \\
			M^{\prime\prime}_{i} &=(1-a_i/a_m)M_{i}, \\
		\end{split}
	\end{equation}
	where $M^{\prime}_{1}=0$ and $M^{\prime\prime}_{m}=0$. Hence, $\mathbf{M}^{\prime}\neq \mathbf{M}$ and $\mathbf{M}^{\prime\prime}\neq \mathbf{M}$. Meanwhile,
	\begin{equation}
		\mathbf{M}=\frac{a_1}{a_1-a_m}\mathbf{M}^{\prime}+\frac{-a_m}{a_1-a_m}\mathbf{M}^{\prime\prime}.
	\end{equation}
	This contradicts to Definition \ref{appdef:ExtremalPOVM}.
	
Next, we prove in the case of rank-1 POVM, the reverse is also true. Assume $\mathbf{M}$ is not extremal and then it can be decomposed to,
	\begin{equation}
		\mathbf{M}=\lambda\mathbf{M}^{\prime}+(1-\lambda)\mathbf{M}^{\prime\prime},
	\end{equation}
with $0<\lambda<1$ and $\mathbf{M}^{\prime}\neq\mathbf{M}^{\prime\prime}$. Since $\mathbf{M}$ is rank-1, it implies that $\forall i$, $M^{\prime}_i=M^{\prime\prime}_i=M_i$, $M^{\prime}_i=0$, or $M^{\prime\prime}_i=0 $. Without loss of generality, assume for $i\le n_1$, $M^{\prime}_i=M^{\prime\prime}_i=M_i$; for $n_1<i\le n_2$, $M^{\prime\prime}_i=0$; and for $i>n_2$, $M^{\prime}_i=0$, where $n_1<n_2<m$. Then, we have
\begin{equation}
\begin{split}
\sum_{i=n_1+1}^{n_2}M_i&=\lambda \sum_{i=n_1+1}^{n_2}M^{\prime}_i=\lambda(\mathbf{1}-\sum_{i=1}^{n_1}M_i),\\
\sum_{i=n_2+1}^{m}M_i&=(1-\lambda) \sum_{i=n_2+1}^{m}M^{\prime\prime}_i=(1-\lambda)(\mathbf{1}-\sum_{i=1}^{n_1}M_i),
\end{split}
\end{equation}
and hence,
\begin{equation}
	\begin{split}
\frac{1}{\lambda}\sum_{i=n_1+1}^{n_2}M_i -\frac{1}{1-\lambda}\sum_{i=n_2+1}^{m}M_i=0
	\end{split}
\end{equation}
which contradicts to the linear independence of $M_i$.
\end{proof}
From the Lemma, we can see that the symmetric and information-complete (SIC) POVM is extremal, which is composed of $ d^{2} $ rank-1 operators, $\ketbra{\phi_{i}}/d$, with normalized vectors $\ket{\phi_{i}}$ satisfying \cite{renes2004symmetric,fuchs2017sic}, $\forall i\neq j$, $\abs{\braket{\phi_{j}}{\phi_{i}}}^{2}=\frac{1}{d+1}$.

Next, we study the intrinsic randomness for extremal POVMs and give a proof of Theorem \ref{appthm:ExtRCan}.

\begin{lemma}\label{applem:purinpsr}
For any two Naimark extensions with pure ancillary state, $\{\mathbf{P}_1,\ketbra{\varphi_1}\}$ and $\{\mathbf{P}_2,\ketbra{\varphi_2}\}$, for arbitrary input state $\rho$, $R(\rho\otimes\ketbra{\varphi_1},\mathbf{P}_1) = R(\rho\otimes\ketbra{\varphi_2},\mathbf{P}_2)$.
\end{lemma}
\begin{proof}

By inserting necessary zero subspace, we can link $\mathbf{P}_1$ and $\mathbf{P}_2$ with a unitary matrix, according to Lemma \ref{applem:UPOVMrepre}. With the freedom to choose a unitary transformation, we can assume $\ket{\varphi_1}=\ket{\varphi_2}=\ket{1}$. By the consistency condition in Eq.~\eqref{appeq:consisMitrQ} and Lemma \ref{applem:Purdecran}, we can obtain $R(\rho\otimes\ketbra{\varphi_1},\mathbf{P}_1) = R(\rho\otimes\ketbra{\varphi_2},\mathbf{P}_2)$.
\end{proof}

\begin{theorem}\label{appthm:ExtRCan}
For an extremal POVM $\mathbf{M}$ and a fixed input state $\rho$, all the generalized Naimark extensions give the same amount of randomness.
\end{theorem}
\begin{proof}
We need to show that the mixture of ancillary states would not affect the intrinsic randomness. Consider a generalized Naimark extension to be $\{\mathbf{P},\sigma\}$.  Without loss of generality, suppose the ancillary state is rank-2 and has the spectral decomposition $\sigma=r\ketbra{\varphi_1}+(1-r)\ketbra{\varphi_2}$ and $0<r<1$. Then, according to Proposition \ref{appPro:CorrDecomp} and Definition \ref{appdef:ExtremalPOVM}, we know that $\{\mathbf{P},\ketbra{\varphi_1}\}$ and $\{\mathbf{P},\ketbra{\varphi_2}\}$ are also Naimark extensions of $\mathbf{M}$.

Consider a pure state $\ket{\varphi}=a\ket{\varphi_1}+b\ket{\varphi_2}$ with $ \abs{a}^2+ \abs{b}^2=1$ and $a,b\in \mathbb{C}$, we have
\begin{equation}
		\begin{split}
			M_i&=\tr_Q[P_i(\mathbf{1}^A\otimes \ketbra{\varphi}^Q)]\\
			&=M_i+ab^*W_i+a^*bW_i^{\dagger},
		\end{split}
\end{equation}
where $ W_i=\tr_Q[P_i(\mathbf{1}^A\otimes \ketbra{\varphi_1}{\varphi_2})] $. The coefficients $a,b$ are arbitrary, thus, $W_i=0$. For any pure states $ \ket{\psi_1} $ and $ \ket{\psi_2} $, we have
	\begin{align}
		\tr(P_i\ketbra{\psi_1\varphi_1}P_i\ketbra{\psi_2\varphi_2})&=\abs{\bra{\psi_1\varphi_1}P_i\ket{\psi_2\varphi_2}}^2\\
		&=\abs{\bra{\psi_1}W_i\ket{\psi_2}}^2=0.
	\end{align}
This implies that the global state $\rho\otimes \sigma= r\rho\otimes\ketbra{\varphi_1}+(1-r)\rho\otimes \ketbra{\varphi_2}$ is always block-diagonal in Definition \ref{appdef:blockdiagP} and according to additivity condition of the randomness functions in Lemma \ref{applem:addblockdiag},
	\begin{equation}
		\begin{split}
			R(\rho\otimes \sigma,\mathbf{P})=rR(\rho\otimes \ketbra{\varphi_1} ,\mathbf{P})+(1-r)R(\rho\otimes \ketbra{\varphi_2} ,\mathbf{P}).
		\end{split}
	\end{equation}
Combine this equation with Lemma \ref{applem:purinpsr}, then the randomness for extremal POVM is given by a canonical Naimark extension
\begin{equation}\label{appeq:RexCanon}
R(\rho,\mathbf{M}) = R(\rho\otimes \ketbra{1},\mathbf{P}_c).
\end{equation}
\end{proof}

\section{Intrinsic randomness for general POVM}\label{App:Intrinsicrandomness}
Here, we prove Theorem \ref{appthm:POVMrandom}, Corollary \ref{appco:iffnonrandom},  Corollary \ref{appco:commoneigen} and Theorem \ref{appth:lbousic}.

\subsection{Intrinsic randomness for general POVM}

\begin{theorem}\label{appthm:POVMrandom}
When Eve performs a measurement on her system $F$, the intrinsic randomness of POVM outcomes is given by,
\begin{equation}\label{appeq:Rexdecomp}
	\begin{split}
R^{cf}(\rho,\mathbf{M}) &= \min_{\{\mathbf{N}^j,r_j\}} \sum_j r_j R(\rho, \mathbf{N}^{j}), \\
\textrm{s.t.} \quad  & \mathbf{M}=\sum_{j} r_j \mathbf{N}^j, \\
	\end{split}
\end{equation}
where the decomposed POVMs $\{\mathbf{N}^{j}\}$ are all extremal and the randomness function $R(\rho, \mathbf{N}^{j})$ is given by Eq.~\eqref{appeq:RexCanon}.
\end{theorem}

\begin{proof}
Consider a generalized Naimark extension $\{\mathbf{P},\sigma\}$. After Eve's measurement, the ancillary state can be treated as a pure state ensemble, $\sigma=\sum_j r_j \ketbra{\varphi_j}$. 
Then, the randomness of $\rho$ is given by
\begin{equation}
	\begin{split}
R^{cf}(\rho,\mathbf{M}) &= \sum_j r_j R(\rho\otimes \ketbra{\varphi_j},\mathbf{P}).
	\end{split}
\end{equation}

The state decomposition corresponds to a POVM decomposition, $\mathbf{M}=\sum_j r_j \mathbf{N}^j$, from Proposition \ref{appPro:CorrDecomp}. To prove the theorem, we only need to prove that all $\mathbf{N}^j$ are extremal, otherwise, we can find a smaller randomness value with a different extension.

Suppose one of the decomposed POVM, $\mathbf{N}^0$, is not extremal and can be decomposed into $\mathbf{N}^0=\lambda \mathbf{N}^{00}+(1-\lambda)\mathbf{N}^{01}, 0<\lambda<1$. Then, there is another decomposition of $\mathbf{M}$,
\begin{equation}
	\begin{split}
\mathbf{M}=\lambda r_0 \mathbf{N}^{00}+(1-\lambda)r_0\mathbf{N}^{01}+\sum_{j\neq 0} r_j \mathbf{N}^j.
	\end{split}
\end{equation}
From Eq.~\eqref{appeq:decompN2P}, we can construct another Nairmark extension $\{\mathbf{P'},\sigma'\}$ with ancillary state,
\begin{equation}
	\begin{split}
\sigma'=\lambda r_0 \ketbra{\varphi_{00}}+(1-\lambda)r_0\ketbra{\varphi_{01}}+\sum_{j\neq 0} r_j \ketbra{\varphi_j'},
	\end{split}
\end{equation}
where $\forall i$, the operator on system $A$, $\bra{\varphi_{00}}P'_i\ket{\varphi_{01}}=0$. After Eve's measurement, the randomness is given by,
\begin{equation}
\begin{split}
& \lambda r_0 R(\rho\otimes \ketbra{\varphi_{00}},\mathbf{P'})+ (1-\lambda)r_0 R(\rho\otimes\ketbra{\varphi_{01}}),\mathbf{P'})+\sum_{j\neq 0} r_j R(\rho\otimes \ketbra{\varphi_j'},\mathbf{P'})\\
&=r_0 R(\rho\otimes (\lambda \ketbra{\varphi_{00}}+ (1-\lambda)\ketbra{\varphi_{01}}),\mathbf{P}')+\sum_{j\neq 0} r_j R(\rho\otimes \ketbra{\varphi_j},\mathbf{P}),
\end{split}
\end{equation}
where the equality comes from the additivity condition of the randomness functions for PVMs and Lemma \ref{applem:purinpsr}.

Now, we only need to prove that
\begin{equation}\label{appeq:exdecran}
\lambda R(\rho\otimes \ketbra{\varphi_{00}},\mathbf{P'})+(1-\lambda)R(\rho\otimes\ketbra{\varphi_{01}}),\mathbf{P'}) \leq R(\rho\otimes \ketbra{\varphi_0},\mathbf{P}),
\end{equation}
for the randomness functions, $R_c$ in Eq.~\eqref{appeq:Rcpvm} and $R_q$ in Eq.~\eqref{appeq:Rqpvm}.

First, we consider the case of $R_c$. For pure input state $\ket{\psi}$, consistency condition implies
\begin{equation}
\lambda \bra{\psi \varphi_{00}}P'_i\ket{\psi \varphi_{00}}+(1-\lambda)\bra{\psi \varphi_{01}}P'_i\ket{\psi \varphi_{01}}=\bra{\psi \varphi_{0}}P_i\ket{\psi \varphi_{0}}.
\end{equation}
Due to the concavity of Shannon entropy, we have
\begin{equation}
\lambda H( \{\bra{\psi \varphi_{00}}P'_i\ket{\psi \varphi_{00}}\})+(1-\lambda)H(\{\bra{\psi \varphi_{01}}P'_i\ket{\psi \varphi_{01}}\})\leq H(\{\bra{\psi \varphi_{0}}P_i\ket{\psi \varphi_{0}}\}).
\end{equation}
Then the Eq.~\eqref{appeq:exdecran} for a general mixed state $\rho$ can be obtained directly.

Now we consider the case of $R_q$. Suppose $\ket{\Phi}^{QF}$ is a purification of state $\tilde{\sigma}=\lambda \ketbra{\varphi_{00}}+ (1-\lambda)\ketbra{\varphi_{01}}$, combine the consistency condition with Lemma \ref{applem:Purdecran} and Lemma \ref{applem:purinpsr}, then
\begin{equation}
S(\rho^{A}\otimes \ketbra{\Phi}^{QF}\parallel\Delta_{\mathbf{P'}^{AQ}\otimes\mathbf{1}^{F}}(\rho^{A}\otimes \ketbra{\Phi}^{QF}))= S(\rho\otimes \ketbra{\varphi_0}\parallel\Delta_{ \mathbf{P}}(\rho\otimes \ketbra{\varphi_0})).
\end{equation}
Take partial trace $\tr_F$ on the l.h.s. of the above equality, from the monotonicity of relative entropy, there is
\begin{equation}
S(\rho \otimes\tilde{\sigma}\parallel\Delta_{\mathbf{P'}}(\rho\otimes \tilde{\sigma}))\leq S(\rho\otimes \ketbra{\varphi_0}\parallel\Delta_{ \mathbf{P}}(\rho\otimes \ketbra{\varphi_0})).
\end{equation}
Combine the inequality with the additivity condition and Eq.~\eqref{appeq:exdecran} holds.

\end{proof}

\subsection{Non-random states}
\begin{corollary}[Necessary and sufficient condition for non-random states]\label{appco:iffnonrandom}
Given a POVM $\mathbf{M}$, a state $\rho$ is non-random, $R^{cf}(\rho,\mathbf{M})=0$, iff the measurement has an extremal decomposition, $\mathbf{M}=\sum_{j}r_j \mathbf{N}^{j}$, satisfying one of the following two equivalent conditions,
\begin{enumerate}
\item
$\forall j, i\neq i'$, $N_{i}^{j}\rho N_{i'}^{j}=0$;

\item
for each $\mathbf{N}^j$, the state has a corresponding spectral decomposition, $\rho=\sum_{k} q_k^j \ketbra{\psi_k^j}$, such that $\forall k$, $N_{i}^{j}\ket{\psi_k^j}=\ket{\psi_k^j}$ for some element $N_{i}^{j}$.
\end{enumerate}
\end{corollary}

\begin{proof}
We first prove that item 2 is a sufficient and necessary condition for $R^{cf}(\rho,\mathbf{M})=0$. 

Sufficiency ($\Rightarrow$):
For each extremal POVM $ \mathbf{N}^j$, there exists a Naimark extension $ \{\mathbf{P}^j,\ketbra{\varphi_j}\} $. According to the condition, for $\forall k$, we can find a POVM element $N_{i}^{j}$ such that
\begin{equation}\label{appeq:eigenpvm}
\begin{split}
\bra{\psi_k^j\varphi_{j}}P_{i}^{j} \ket{\psi_k^j\varphi_{j}} &=\tr[\ketbra{\psi_k^j}^{A}\tr_Q (P_{i}^{j}(\mathbf{1}^{A}\otimes\ketbra{\varphi_{j}}^{Q}))]\\
  &=\tr(N_{i}^{j}\ketbra{\psi_k^j})=1.
  \end{split}
  \end{equation}  
It follows that $P_{i}^{j}\ket{\psi_k^j \varphi_j}=\ket{\psi_k^j \varphi_j}$, and moreover, $R^{cf}(\rho,\mathbf{M})=\sum_j r_j R(\rho,\mathbf{N}^j)=0$.

Necessity ($\Leftarrow$):
If $R^{cf}(\rho,\mathbf{M})=0$, there exists an extremal decomposition $ \mathbf{M}=\sum_{j}r_j \mathbf{N}^{j} $, such that $\forall j$, $ R(\rho,\mathbf{N}^j)=R(\rho\otimes \ketbra{\varphi_j},\mathbf{P}^j)=0$. Therefore, $\rho\otimes \ketbra{\varphi_j}$ does not change after the associated block-dephasing operation. We write the dephased state into its spectral decomposition,
\begin{equation}
\rho\otimes \ketbra{\varphi_j}=\sum_i P_i^j(\rho\otimes \ketbra{\varphi_j})P_i^j=\sum_{k} q_{k}^j \ketbra{u_{k}^j}.
\end{equation}
The eigenstate of $ \rho\otimes \ketbra{\varphi_j} $ must be product state. Thus, $ \ket{u_{k}^j}= \ket{\psi_{k}^j \varphi_j}$ is also an eigenstate of some PVM element $P_{i}^j$. From Eq.~\eqref{appeq:eigenpvm}, we can obtain $N_{i}^{j}\ket{\psi_k^j}=\ket{\psi_k^j}$.

The sufficiency and necessity of item 1 can be deduced from the fact that  $N_{i}^{j}\rho N_{i'}^{j}=0, i\neq i'$ is sufficient and necessary conditions for $R(\rho\otimes\ketbra{\varphi_j},\mathbf{P}^j)=0$ \cite{Bischof2019}.

\end{proof}

To prove Corollary \ref{appco:commoneigen}, we need the following definition and two lemmas.
\begin{definition} [Grouping/Coarse graining process \cite{Haapasalo2012}]
 Give two POVMs $ \mathbf{M}=\{M_{1},\cdots,M_{m}\}$ and $ \mathbf{N}=\{N_{1},\cdots,N_{n}\}$ with $ m\leq n$. A grouping is determined by a mapping $ f:[n] \rightarrow [m]$. Concretely, the element of two POVMs have the relation  
	\begin{equation}
		M_{i} =\sum_{j\in f^{-1}(i)}N_{j},	
	\end{equation} 
	where $f^{-1}$ represents the inverse mapping.
\end{definition}
\begin{lemma}\label{applem:extdecom}
For a POVM, $ \mathbf{M} =\{M_{1},\cdots,M_{m}\}$, if $\ket{\psi_1},\cdots,\ket{\psi_l}$ are pairwise orthogonal states and each $\ket{\psi_k}$ is a common eigenstate of all elements $M_i$, then there exists a decomposition 
	\begin{equation}
		\mathbf{M}=\sum_{j}r_j\mathbf{N}^j,
	\end{equation}
such that for arbitrary pure state $\ket{\psi_k}$, each $\mathbf{N}^j$ contains an element $ N_{i}^j$ satisfying $ N_{i}^j\ket{\psi_k}=\ket{\psi_k}$.
\end{lemma}
\begin{proof}
From the condition, each POVM element has an eigendecomposition
	\begin{equation}
		M_i=\sum_{i'=1}^{d} \lambda_{ii'}\ketbra{\psi_{ii'}},
	\end{equation}
	where $ \ket{\psi_{i1}} =\ket{\psi_{1}},\cdots, \ket{\psi_{il}} =\ket{\psi_{l}}$ and $\lambda_{ii'}\ge 0, \sum_{i'}\lambda_{ii'}=1.$ All rank-1 terms form a new POVM denoted as $\mathbf{A}=\{\lambda_{ii'}\ketbra{\psi_{ii'}},i\in[m],i'\in[d] \} $. Via the method that decomposes a POVM into a mixture of extremal POVMs in the proof of Lemma \ref{applem:linIndExt}, the rank-1 POVM $\mathbf{A}$ can be decomposed into a mixture of extremals. From the description of the algorithm, the elements of extremal POVMs inherit the elements of the POVM $\mathbf{A}$ with only a difference in coefficients. As a consequence, each extremal POVM has elements $  \{c_1|\psi_1\rangle \langle \psi_1|,\cdots,c_l|\psi_l\rangle \langle \psi_l|\}$ and other elements are in an orthogonal space, thus, $ c_1=\cdots=c_l=1. $ Write the decomposition as 
	\begin{equation}
		\mathbf{A}=\sum_{j}r_j \mathbf{B}^{j},
	\end{equation}
 and each POVM $\mathbf{B}^j=\{B^j_{ii'},i\in[m],i'\in[d]\}$ contains elements $ \{|\psi_1\rangle \langle \psi_1|,\cdots,|\psi_l\rangle \langle \psi_l|\}$. Define $ \mathbf{N}^j =\{N^{j}_1,\cdots,N^{j}_m\}$ as a grouping of $ \mathbf{B}^j$,
	\begin{equation}
		N^{j}_i=\sum_{i'=1}^{d}B_{ii'}^j.
	\end{equation}
Therefore, $ \mathbf{M}=\sum_{j}r_j \mathbf{N}^{j}$, in addition, for each $ \mathbf{N}^{j}$  and each $ \ket{\psi_k}$, there exists a related $N^{j}_{i} $ such that $N^{j}_{i}\ket{\psi_k}=\ket{\psi_k}$.
\end{proof}	
	
\begin{lemma}\label{applem:necenonran}
If $ R(\rho,\mathbf{M}) =0$, then $ [\rho,M_i]=\rho M_i-M_i\rho=0 ,\forall i$.
\end{lemma}
\begin{proof}
Since $ R(\rho,\mathbf{M}) =0$, there exists a Naimark extension $\{\mathbf{P},\sigma\}$ such that the global input state has a spectral decomposition
\begin{equation}\label{appent}
\rho\otimes \sigma=\Delta_{\mathbf{P}}(\rho\otimes\sigma)=\sum_{k}\lambda_{k}\ketbra{u_k},
\end{equation}
where $\forall i$, $P_i\ket{u_k}=\ket{u_k}$, or $P_i\ket{u_k}=0$. In either case, we can obtain $P_i(\rho\otimes\sigma )=(\rho\otimes\sigma )P_i$. Take partial trace $\tr_Q $ on both sides of the equality, 
\begin{equation}
\tr_Q [P_i (\mathbf{1}^A \otimes \sigma^Q)(\rho^A \otimes \mathbf{1}^Q)]=\tr_Q [(\rho^A\otimes \mathbf{1}^Q)(\mathbf{1}^A \otimes \sigma^Q)P_i].
\end{equation}
Combined with the consistency condition, there is $ M_i\rho=\rho M_i. $
\end{proof}
Now, we give proof of Corollary \ref{appco:commoneigen}.
\begin{corollary}\label{appco:commoneigen}
Given a POVM $\mathbf{M}$, a pure state $\ket{\psi}$ is non-random, $R(\ketbra{\psi},\mathbf{M})=0$, iff $\ket{\psi}$ is a common eigenstate of all measurement elements.
\end{corollary}

\begin{proof}
From Lemma \ref{applem:necenonran}, the sufficiency is apparent. The two randomness functions have the relation $ R(\rho,\mathbf{M})\le R^{cf}(\rho,\mathbf{M})$, therefore, the necessity can be directly obtained by Corollary \ref{appco:iffnonrandom} and Lemma \ref{applem:extdecom}.
\end{proof}

\subsection{Lower bound for intrinsic randomness of SIC measurement}
\begin{theorem}\label{appth:lbousic}
For a SIC measurement $\mathbf{M}$, a lower bound of intrinsic randomness is given by,
\begin{equation}
R(\mathbf{M})>\log(\frac{d+1}{2}),
\end{equation}
where $d$ is the dimension of the corresponding space.
\end{theorem}

\begin{proof}

For SIC POVM $\mathbf{M}=\{\ketbra{\phi_{1}}/d,\cdots,\ketbra{\phi_{d^2}}/d\}$, let $\{\mathbf{P},\ketbra{1}\}$ be a canonical Naimark PVM.

(i) When $R=R_c$, the minimum randomness is achieved by some pure states from the expression Eq.~\eqref{appeq:Rcpvm}. The randomness for pure state $\ket{\psi}$ is equal to  Shannon entropy $ H(\{p_i\})$ with $p_i= \abs{\braket{\phi_{i}}{\psi}}^2/d$. For any state $\rho$ and $ p_i=\tr(\rho \ketbra{\phi_{1}}/d)$, the entropy of the distribution has a state-independent lower bound \cite{rastegin2013uncertainty},
\begin{equation}\label{appeq:lowbousic}
  H(\{p_i\})\ge \log(\frac{d(d+1)}{\tr(\rho^2)+1}).
  \end{equation}  
As a result,
\begin{equation}
R_{c}(\mathbf{M})\ge \log(\frac{d(d+1)}{2})> \log(\frac{d+1}{2}).
\end{equation}

(ii) When $R=R_q(\rho\otimes\ketbra{1},\mathbf{P})=S(\Delta_{\mathbf{P}}(\rho\otimes\ketbra{1}))-S(\rho)$, assume $ P_i(\rho\otimes\ketbra{1})P_i=\sum_{k} q_{ik}\ketbra{u_{ik}}$ is a spectral decomposition and $\sum_{k} q_{ik}=\tr(\rho \ketbra{\phi_{1}}/d)=p_i $. The entropy of probability distribution $ \{q_{ik}/p_i\} $ is 
\begin{equation}
0\le -\sum_k\frac{q_{ik}}{p_i}\log(\frac{q_{ik}}{p_i})=-\frac{1}{p_i}[\sum_k (q_{ik}\log q_{ik})- p_i\log p_i ].
\end{equation}
Hence, $-\sum_k q_{ik}\log q_{ik}\ge -p_i\log p_i$. Then the term $ S(\Delta_{\mathbf{P}}(\rho\otimes\ketbra{1}))=S[\sum_iP_i(\rho\otimes\ketbra{1})P_i]\ge H(\{p_i\})$ and equality holds only if $\rho$ is pure. Combine this with Eq.~\eqref{appeq:lowbousic} and there is 
\begin{equation}
R_q(\rho,\mathbf{M})\ge H(\{p_i\})-S(\rho)\ge \log(\frac{d(d+1)}{\tr(\rho^2)+1})-S(\rho)>\log(\frac{d+1}{2}).
 \end{equation} 
Then, $R_{q}(\mathbf{M})> \log((d+1)/2)$.
\end{proof}

Remark: 
When $R=R_q$, to evaluate how tight our lower bound is, let us examine the special case of the maximally mixed input state. The amount of randomness from a general POVM has an upper bound
\begin{equation}
R_q(\mathbf{1}/d,\mathbf{M})=S(\Delta_{\mathbf{P}}(\mathbf{1}/d\otimes\ketbra{1}))-S(\rho)\le \log d^2-\log d=\log d.
\end{equation}
This indicates $R_{q}(\mathbf{M})\le \log d$. The difference between our lower bound $\log((d+1)/2)$ and this upper bound is less than $1$ bit.

\subsection{Numerical evaluation}
Beyond special cases such as SIC measurements, we present a numerical approach to evaluate intrinsic randomness under a general POVM. For the convex-roof-type randomness measure defined by the optimization in Eq.~\eqref{appeq:Rexdecomp}, the objective function is biconvex in its arguments, the probability distribution $\{r_j\}$ and decomposed extremal POVMs $\{\mathbf{N}^j\}$. Also, the constraints form a convex set. For such optimization problems, a global optimal value can be obtained in principle~\cite{gorski2007biconvex}. Nevertheless, there are two difficulties: (1) the characterization of the set of extremal POVMs is relatively complex~\cite{d2005classical}, (2) the dimension of the probability distribution $\{r_j\}$ is not fixed.

To tackle the first problem, we can remove the constraint that $\{\mathbf{N}^j\}$ are all extremal. To see why this is the case, suppose the optimal value to Eq.~\eqref{appeq:Rexdecomp} is given by the tuple $\{r_j^{*},\mathbf{N}^{*j}\}$. In Theorem~\ref{appthm:POVMrandom}, we show that for any probability distribution $\{\tilde{r}_k\}$ and POVMs $\{\tilde{\mathbf{N}}^k\}$ such that $\mathbf{M}=\sum_k\tilde{r}_k\tilde{\mathbf{N}}^k$, the following inequality holds,
\begin{equation}
  \sum_k\tilde{r}_k R\left(\rho\otimes\ketbra{0},\tilde{\mathbf{P}^k}\right)\geq\sum_jr_j^{*}R\left(\rho\otimes\ketbra{0},\mathbf{P}^{*j}\right) =\sum_jr_j^{*}R\left(\rho,\mathbf{N}^{*j}\right),
\end{equation}
where $\tilde{\mathbf{P}^k},\mathbf{P}^{*j}$ represent the canonical Naimark extension of $\tilde{\mathbf{N}}^k,\mathbf{N}^j$, respectively. With a slight abuse of notation, we also write $R\left(\rho\otimes\ketbra{0},\tilde{\mathbf{P}^k}\right)$ as $R\left(\rho,\tilde{\mathbf{N}^k}\right)$ for a general POVM $\tilde{\mathbf{N}^k}$. Therefore, the optimization in Eq.~\eqref{appeq:Rexdecomp} is equivalent to the following problem,
\begin{equation}\label{appEq:EquivOpt}
	\begin{split}
R^{cf}(\rho,\mathbf{M}) &= \min_{\{\mathbf{N}^j,r_j\}} \sum_j r_j R(\rho, \mathbf{N}^{j}), \\
\textrm{s.t.} \quad   \mathbf{M}&=\sum_{j} r_j \mathbf{N}^j, \\
\mathbf{N}^j&\in \mathcal{P} ,\forall j, \\
r_j&\geq0,\forall j, \\
\sum_jr_j&=1,
	\end{split}
\end{equation}
where $ \mathcal{P}$ is the set of POVMs.

To tackle the second problem that the dimension of $\{r_j\}$ is not fixed, we can apply Carath\'eodory's theorem for convex hulls. Suppose $\mathbf{M}$ is a POVM acting on a $d$-dimensional Hilbert space with $m$ elements. After considering the positive-semidefinite property and completeness, it can be parameterized by $(md^2-1)$ real parameters. Then, according to Carath\'eodory's theorem, the optimal value to Eq.~\eqref{appEq:EquivOpt} can be attained by a probability distribution $\{r_j\}$ with at most $md^2$ terms. Consequently, we can restrict the dimension of $\{r_j\}$ to be $md^2$ in Eq.~\eqref{appEq:EquivOpt} without loss of generality.

With the above results, the global optimum to Eq.~\eqref{appEq:EquivOpt} can now be efficiently solved numerically. In particular, for a fixed probability distribution $\{r_j\}$, the problem becomes a semi-definite programming optimization in the arguments $\{\textbf{N}^j\}$. 
\end{appendix}

\end{document}